\documentclass[aps,reprint,11pt,tightenlines,secnumarabic,nobibnotes,longbibliography,notitlepage,floatfix,onecolumn,superscriptaddress,nofootinbib]{revtex4-2}
\usepackage{color, xcolor, colortbl}
\usepackage{graphicx}
\usepackage[percent]{overpic}
\usepackage{amsmath,amssymb,amsthm,amsfonts,mathrsfs}
\usepackage{dsfont,pifont}
\usepackage{mathtools}
\usepackage{microtype}
\usepackage[ruled,vlined,linesnumbered]{algorithm2e}
\usepackage{algpseudocode}
\usepackage{dcolumn}
\usepackage{epstopdf}
\usepackage{bm}
\usepackage{appendix}
\usepackage{multirow}
\usepackage{braket}
\usepackage[english]{babel}
\usepackage[colorlinks=true,linkcolor=black,citecolor=black,urlcolor=black]{hyperref}
\usepackage[capitalize]{cleveref}
\usepackage{aliascnt}
\usepackage[T1]{fontenc}
\usepackage{xpatch}
\usepackage{adjustbox}
\usepackage{xspace}
\usepackage[roman]{complexity}
\usepackage{xparse}
\usepackage{comment}
\usepackage{MnSymbol}

\newcommand{\cmark}{\textcolor{green}{\ding{51}}}
\newcommand{\xmark}{\textcolor{red}{\ding{55}}}

\renewcommand{\d}{\mathrm{d}}

\newcommand{\Tr}{\operatorname{Tr}}

\newcommand{\abs}[1]{\left\lvert#1\right\rvert}
\newcommand{\norm}[1]{\left\lVert#1\right\rVert}
\newcommand{\Or}{\mathcal{O}}

\newcommand{\ud}{\,\mathrm{d}}
\newcommand{\BB}{\mathbb{B}}

\newcommand{\RR}{\mathbb{R}}
\newcommand{\CC}{\mathbb{C}}

\newcommand{\cC}{\mathcal{C}}
\newcommand{\cD}{\mathcal{D}}
\renewcommand{\cL}{\mathcal{L}}

\newcommand{\cH}{\mathcal{H}}

\newcommand{\cE}{\mathcal{E}}

\newcommand{\cO}{\mathcal{O}}
\newcommand{\cA}{\mathcal{A}}
\newcommand{\cB}{\mathcal{B}}
\newcommand{\cW}{\mathcal{W}}
\newcommand{\Hcl}{H_{\rm cl}}
\newcommand{\rhocl}{\rho_{\rm cl}}

\newcommand{\mdr}{\Delta_{\sigma}}
\newcommand{\wgt}{\Gamma_{\sigma}}
\newcommand{\sr}{\sigma^{1/2}}

\newtheorem{thm}{\protect\theoremname}
\theoremstyle{plain}

\theoremstyle{plain}
\newaliascnt{lem}{thm}
\newtheorem{lem}[lem]{\protect\lemmaname}
\aliascntresetthe{lem}
\theoremstyle{plain}
\newtheorem{rem}[thm]{\protect\remarkname}
\theoremstyle{plain}
\newtheorem*{lem*}{\protect\lemmaname}
\theoremstyle{plain}
\newaliascnt{prop}{thm}
\newtheorem{prop}[prop]{\protect\propositionname}
\aliascntresetthe{prop}
\theoremstyle{plain}

\newtheorem{defn}[thm]{\protect\definitionname}

\theoremstyle{definition}

\providecommand{\definitionname}{Definition}
\providecommand{\assumptionname}{Assumption}
\providecommand{\corollaryname}{Corollary}
\providecommand{\lemmaname}{Lemma}
\providecommand{\propositionname}{Proposition}
\providecommand{\remarkname}{Remark}
\providecommand{\examplename}{Example}
\providecommand{\theoremname}{Theorem}
\providecommand{\conjecturename}{Conjecture}
\providecommand{\assumptionname}{Assumption}

\Crefname{lem}{Lemma}{Lemmas}
\Crefname{prop}{Proposition}{Propositions}

\newcommand{\ketbra}[2]{\mathinner{|{#1}\rangle\!\langle{#2}|}}

\renewcommand{\rrangle}{\rangle\!\rangle}
\renewcommand{\llangle}{\langle\!\langle}

\newcommand{\prlsection}[1]{\section{#1}}

\newcommand{\DeptMath}{Department of Mathematics, University of California, Berkeley, CA 94720, USA}
\newcommand{\LBLMath}{Applied Mathematics and Computational Research Division, Lawrence Berkeley National Laboratory, Berkeley, CA 94720, USA}
\newcommand{\Simons}{Simons Institute for the Theory of Computing, University of California, Berkeley, CA 94720, USA}
\newcommand{\VT}{Institute for Advanced Computing and Department of Computer Science, Virginia Tech, Alexandria, VA 22305, USA}
\newcommand{\Caltech}{Department of Computing and Mathematical Sciences, California Institute of Technology, Pasadena, CA 91125, USA}

\begin{document}

\title{Accelerating quantum Gibbs sampling without quantum walks}

\affiliation{\VT}
\affiliation{\Simons}
\affiliation{\Caltech}
\affiliation{\DeptMath}
\affiliation{\LBLMath}

\author{Jiaqi Leng}
\thanks{These authors contributed equally to this work.}
\affiliation{\VT}
\affiliation{\Simons}
\affiliation{\DeptMath}
\author{Jiaqing Jiang}
\thanks{These authors contributed equally to this work.}
\affiliation{\Simons}
\author{Lin Lin}
\thanks{lin@caltech.edu}
\affiliation{\Caltech}
\affiliation{\DeptMath}
\affiliation{\LBLMath}
\date{\today}

\begin{abstract}
Szegedy's quantum walk quadratically accelerates reversible classical Markov chains, but extending this mechanism to efficiently implementable quantum Gibbs samplers has remained challenging beyond special classes. We present a walk-free quantum algorithm for preparing purified Gibbs states with a quadratic improvement in spectral-gap dependence for a broad class of quantum Gibbs samplers that satisfy exact Kubo--Martin--Schwinger detailed balance. Our main structural result is an explicit factorization of the corresponding parent Hamiltonian into noncommutative first-order operators. This turns purified Gibbs-state preparation into a singular-value filtering problem and enables a quantum singular value transformation algorithm with quadratically improved gap dependence under standard coherent-access and warm-start assumptions. The framework applies to several efficiently implementable Gibbs samplers beyond the Davies setting. 
Moreover, we construct noncommuting Hamiltonians with efficiently preparable warm starts whose refinement under quantum Gibbs samplers requires exponential time. For these models, our algorithm achieves an end-to-end quantum speedup over the worst-case mixing scale.
\end{abstract}
\maketitle

\prlsection{Introduction}\label{sec:intro}

Gibbs sampling is a fundamental primitive in statistical physics~\cite{metropolis1953equation,glauber1963time,temme2011quantum}, computational chemistry~\cite{panagiotopoulos1987direct,ceperley1995path,FrenkelSmit2002}, statistics~\cite{geman1984stochastic,gelfand1990sampling,casella1992explaining}, machine learning~\cite{hinton2006reducing,hinton2010practical,neal2012bayesian}, and combinatorial sampling~\cite{kirkpatrick1983optimization,dyer1991random,levin2017markov}. The efficiency of Gibbs sampling is governed by the inverse of its spectral gap and is limited when this gap is small~\cite{chayes1987exponential,cesi1996two,schonmann1987second}.
Szegedy's quantum walk~\cite{Szegedy2004} provides a general mechanism to accelerate Gibbs sampling,  and more broadly, reversible Markov chains using quantum algorithms. This framework has become a central paradigm in quantum algorithm design, enabling quadratic speedups in applications such as quantum search~\cite{shenvi2003quantum,magniez2007search,buhrman2004quantum}. Specifically, for a reversible Markov chain with spectral gap $\Delta$, Szegedy's quantum walk lifted it to a unitary walk with phase gap of order $\sqrt{\Delta}$, yielding a quadratic improvement in the gap dependence under coherent access assumptions.

A natural question is whether a similar quadratic speedup can be achieved for quantum Gibbs sampling~\cite{chen2025efficient,ding2025efficient,jiang2024quantum,gilyen2024quantum,ding2025end,temme2011quantum}. Quantum Gibbs sampling is the quantum analogue of Markov chain Monte Carlo and is primarily designed for simulating finite-temperature quantum many-body systems, with applications also in quantum optimization~\cite{brandao2017quantum,van2017quantum} and quantum machine learning~\cite{wilde2025fundamentals,amin2018quantum}.
In the Lindblad framework, a quantum Gibbs sampler is generated by a Lindbladian $\cL$ whose stationary state is the Gibbs state of a many-body Hamiltonian. When $\cL$ satisfies the Kubo--Martin--Schwinger (KMS) detailed balance condition, it is the natural quantum generalization of a reversible Markov generator. 
For the special case of Davies generators~\cite{davies1974markovian}, Wocjan and Temme obtained a Szegedy-type quadratic speedup~\cite{WocjanTemme2023} relying on the special frequency-resolved structure of Davies generators; see also \cite{chen2023quantum} for Szegedy-type constructions for Lindbladians that approximate Davies generators.  
However, no comparable quantum walk construction is known for other classes of efficiently implementable quantum Gibbs samplers satisfying exact KMS detailed balance~\cite{chen2025efficient,ding2025efficient,gilyen2024quantum,chen2023efficient}. These exact KMS-detailed-balanced samplers include an additional coherent term, which makes it difficult to extend Szegedy-type techniques~\cite{chen2023efficient,ding2025efficient}. A recent work~\cite{di2026quantum} uses the detectability lemma and achieves a quadratic improvement in gap dependence for local commuting Hamiltonians. It remains an open problem whether one can still obtain a quadratic speedup using efficiently implementable quantum Gibbs samplers~\cite{chen2025efficient,ding2025efficient,gilyen2024quantum,chen2023efficient} for general non-commuting Hamiltonians. 

\begin{table*}
\caption{\label{tab:table3}
Comparison of existing methods to achieve quadratic speedups in spectral gap dependence for classical and quantum Markov chains.
``Efficient implementation'' indicates whether there exists polynomial-time quantum implementation of the method. }
\begin{tabular*}{\textwidth}{@{\extracolsep{\fill}}l|llc@{}}
 \hline
 \hline
 Method & Hamiltonian model & Gibbs sampler & Efficient \\ 
        &                   &               & implementation \\ \hline
 Szegedy walk~\cite{Szegedy2004,somma2007quantum} & Classical & MCMC & \cmark \\
 Witten Laplacian factorization~\cite{leng2026operator} & Classical & MCMC & \cmark \\ \hline
 Szegedy walk~\cite{WocjanTemme2023,chen2023quantum} & Quantum, general & (Approximate) Davies & \xmark \\ 
 Detectability lemma~\cite{di2026quantum} & Quantum, local commuting & KMS & \cmark \\
 \textbf{This work} & Quantum, general & KMS & \cmark \\
 \hline
 \hline
\end{tabular*}
\end{table*}

In this work, we answer this question affirmatively. Our key observation is that a broad class of KMS-detailed-balanced Lindbladians can be written in a unified Kossakowski-matrix form, encompassing Davies generators and the efficiently implementable quantum Gibbs samplers of Refs.~\cite{chen2025efficient,ding2025efficient}. 
For this class, the Lindbladian $\cL$ can be mapped, via a Gibbs-weighted similarity transformation, to a parent-Hamiltonian super-operator $\cH$. After vectorization, $\cH$ becomes a Hermitian operator $\widehat{\cH}$ on the doubled Hilbert space, and the canonical purified Gibbs state (also known as the thermofield double state, or TFD state) is a ground state of $\widehat{\cH}$. 
Moreover, we show that the Kossakowski structure immediately yields a factorization of the vectorized parent Hamiltonian:
\begin{equation}\label{eq:factorization}
\widehat{\cH} = \BB^\dagger \BB,
\end{equation}
where $\BB$ is a noncommutative analogue of the gradient operator that maps states from the doubled Hilbert space to an enlarged space.
Whenever the Gibbs state is the unique stationary state of the Lindbladian $\cL$, the kernel of $\BB$ is of rank 1 and is spanned by the purified Gibbs state.

This factorization is the mechanism behind our quadratic improvement in gap dependence. Instead of constructing a Szegedy-type operator, we work directly with $\BB$. Its smallest nonzero singular value is $\sqrt{\Delta}$, where $\Delta$ is the gap of the Lindbladian $\cL$. Quantum singular value transformation (QSVT)~\cite{GilyenSuLowEtAl2019} therefore yields a purified-Gibbs-state preparation algorithm with gap dependence $\Or(\Delta^{-1/2}\log(1/\epsilon))$, assuming coherent access to the first-order factors and a warm start, and hence gives a walk-free quadratic improvement. 
The output of our quantum algorithm is the purified Gibbs state, which can be used to estimate thermal observables via amplitude estimation~\cite{brassard2000quantum,rall2021faster,rall2023amplitude}.
Conceptually, this parallels the recent operator-level acceleration for classical continuous-space Gibbs sampling based on factorizing the Witten-Laplacian~\cite{leng2026operator}. 
Although the existence of related factorizations for general KMS-symmetric generators is known in an abstract setting~\cite{cipriani2003,carlen2020a,vernooij2023derivations}, our construction is explicit and, for several prominent classes of quantum Gibbs samplers, efficiently implementable from a concrete positive-matrix description. 
A comparison of our algorithm with existing methods is available in~\cref{tab:table3}.

Like Szegedy-type algorithms~\cite{Szegedy2004,WocjanTemme2023}, QSVT-based filtering methods typically require an initial state (i.e., warm start) with significant overlap with the target purified Gibbs state. 
Preparing such a state can itself be costly, so realizing the quadratic speedup requires controlling this initial cost.
We add weak, inhomogeneous $XX$ interactions to the classical Curie--Weiss Hamiltonian and call the resulting quantum systems \emph{inhomogeneous Curie--Weiss models}. At sufficiently low fixed temperature, a constant-depth circuit prepares a warm start with squared overlap $>0.99$.
Nevertheless, the CKG sampler~\cite{chen2023efficient} with local couplings requires time exponential in system size to refine the same warm start below a smaller, system-size-independent error threshold.
Easy warm-start preparation thus coexists with slow relaxation from that same warm start. These models provide a concrete setting in which the warm-start requirement is efficiently satisfied and our quadratic improvement in gap dependence can be applied.

\prlsection{Theory}\label{sec:theory}

\subsection{Quantum Gibbs samplers and KMS detailed balance}
Given an $n$-qubit Hamiltonian $H$ and an inverse temperature $\beta > 0$, the Gibbs state is $\sigma = Z_\beta^{-1} e^{-\beta H}$, where $Z_\beta = \Tr[e^{-\beta H}]$. We can write the Gibbs state $\sigma = \sum_j \sigma_j \ketbra{\psi_j}{\psi_j}$, where $(E_j,\ket{\psi_j})_j$ are the eigen-pairs of the Hamiltonian $H$, and $\sigma_j = e^{-\beta E_j}/Z_\beta$.
Define the KMS inner product $\langle X,Y\rangle_{\sigma} \coloneqq \Tr(X^\dagger \sigma^{1/2} Y \sigma^{1/2})$. If a Lindbladian $\cL$ satisfies the KMS detailed balance condition with respect to $\sigma$, i.e., $\langle X, \cL^{\dag}(Y)\rangle_{\sigma} = \langle \cL^{\dag}(X), Y\rangle_{\sigma}$, then $\sigma$ is a stationary state of $\cL$, i.e., $\cL[\sigma]=0$.
Therefore, KMS-detailed-balanced Lindbladians are natural choices of quantum Gibbs samplers.
We consider KMS-detailed-balanced Lindbladians of the form:
\begin{align}\label{eq:kossakowski-form}
\cL[\cdot]=-i[G,\cdot]+\sum_{\nu,\nu'}\alpha_{\nu,\nu'}\left(A_\nu\cdot A_{\nu'}^\dagger-\frac{1}{2}\{A_{\nu'}^\dagger A_\nu,\cdot\}\right),
\end{align}
where $\nu, \nu' \in B_H$ range over the Bohr frequencies of $H$, $A_\nu$ denotes the corresponding Bohr-frequency component of one or more coupling operators, and $\mathbf{C} = (\alpha_{\nu,\nu'})_{\nu,\nu'\in B_H}\succeq 0$ is the \textit{Kossakowski} matrix. The KMS detailed balance condition uniquely determines the coherent term $G$ once $\mathbf{C}$ is specified~\cite{ding2025efficient}.
This single template~\cref{eq:kossakowski-form} covers important classes of KMS-detailed-balanced Lindbladians, including Davies generators and the efficiently implementable quantum Gibbs samplers recently proposed in~\cite{chen2023efficient,ding2025efficient}; they differ only in how $\mathbf{C}$ decomposes into positive pieces.
For details, see~\cref{append:factorization-KMS-lindbladian}.

\subsection{Szegedy's quantum walk and obstruction to quantum Gibbs sampling}
For classical Markov chains, Szegedy's quantum walk is built from the following mechanism. 
Let $P$ be an ergodic reversible transition matrix where $P_{xy}$ denotes the probability of transiting from $x$ to $y$. Denote the invariant probability distribution of $P$ as $\pi$. By a similarity transform, $P$ is mapped to a real symmetric discriminant matrix $D$ with the same eigenvalues, where $D:= \operatorname{diag}(\pi)^{\frac{1}{2}} P\, \operatorname{diag}(\pi)^{-\frac{1}{2}}.$  
Szegedy's~\cite{Szegedy2004} key step is to represent this discriminant as $D=T^{\dag}RT$ (where $T$ is an isometry and $R$ is a reflection), and define a walk unitary $\cW = R(2\Pi -I)$, where $\Pi$ denotes the projector onto the image of $T$.
The walk unitary $\cW$ maps each eigenvalue $\lambda_i$ of $D$ to phases $e^{\pm i\arccos(\lambda_i)}$, so a spectral gap $\Delta=1-\lambda_2$ is converted into a phase gap $\arccos(1-\Delta)\approx\sqrt{2\Delta}$.
For quantum Gibbs samplers, one would like an analogous construction with $P$ replaced by a quantum channel.

To achieve a quadratic speedup for Markov chains, we need an efficient implementation of the isometry $T$.
In classical settings where the transition probabilities are computable, such an isometry can be constructed via coherent arithmetic.
For a general KMS-detailed-balanced Lindbladian, no analogous efficiently implementable isometry is known. 
Ref.~\cite{WocjanTemme2023} circumvents this difficulty for Davies generators, where the frequency-resolved structure provides the missing ingredients. Even in that case, however, implementing $T$ requires resolving all Hamiltonian eigenvalues under the \textit{rounding promise}~\cite{WocjanTemme2023}, an assumption widely regarded as highly unphysical. 
For the broader classes of quantum Gibbs samplers considered in~\cite{chen2025efficient,ding2025efficient}, no comparable Szegedy-type isometry is known, and correspondingly no efficient quantum implementation is available.

\subsection{Sum-of-squares factorization of Lindbladians in Kossakowski form}

If a Lindbladian $\cL$ satisfies the KMS detailed balance, we define its parent Hamiltonian:
\begin{equation}
  \cH[\cdot] = -\sigma^{-1/4}\cL[\sigma^{1/4}\cdot\sigma^{1/4}]\sigma^{-1/4}, \label{eq:cH}
\end{equation}
which is a Hermitian positive-semidefinite operator with the same spectrum as $-\cL$, and $\cH[\sigma^{1/2}]=0$. 

Our algorithm bypasses the missing Szegedy isometry by directly factorizing the parent Hamiltonian. First, we introduce the following symmetric bilinear form associated with $\cL$:
\begin{equation}
\cE(X,Y) \coloneqq -\langle X,\cL^{\dagger}(Y)\rangle_{\sigma}.
\end{equation}
Note that $\cE(X,X)$ is known as the \textit{Dirichlet form} of $\cL$.
For Lindbladian generators of the form~\cref{eq:kossakowski-form}, we show that a positive decomposition of the Kossakowski matrix yields $\cE$ as a weighted sum of squared commutators,
\begin{equation}
\cE(X,Y)=\int_{-\infty}^{\infty} g(t)\sum_{j=1}^J \langle [B_{j,t},X],[B_{j,t},Y]\rangle_{\sigma}\,\ud t,
\label{eq:bilinear-sos}
\end{equation}
where $g(t)=1/[\beta\cosh(2\pi t/\beta)]$, $J$ is the rank of $\mathbf{C}$, $B_{j,t}=e^{i Ht}B_je^{-i Ht}$, and each $B_{j}$ is a fixed linear combination of $\{A_\nu\}$ determined by a positive decomposition of $\mathbf{C}$; see~\cref{append:symmetric-form-kossakowski}.
The factorized Dirichlet form~\cref{eq:bilinear-sos} was previously known only in restricted settings~\cite{chen2025quantum,rouze2025efficient,basso2025quantum,slezak2026polynomial}. It is the quantum analogue of writing a reversible diffusion generator as an integral of squared derivatives, with dissipation resolved into noncommutative gradients labeled by $j$ and by the time parameter $t$.

This representation immediately lifts to the parent Hamiltonian. For each $B_j$ as discussed above, we define
\begin{equation}\label{eqn:B_sharp_B_flat}
B_{j}^{\sharp}=\sigma^{1/4}B_j\sigma^{-1/4},\quad
B_{j}^{\flat}=\sigma^{-1/4}B_j\sigma^{1/4}.
\end{equation}
We write $B^*_{j,t} := e^{iHt} B^*_j e^{-iHt}$ for $* \in \{\sharp,\flat\}$, and define $\cB_{j,t}[X] \coloneqq B_{j,t}^{\sharp}X-XB_{j,t}^{\flat}$. From~\cref{eq:bilinear-sos}, we show that 
\begin{equation} 
\cH[\cdot]=\int_{-\infty}^{\infty} g(t)\sum_j \cB_{j,t}^{\dagger}\circ\cB_{j,t}[\cdot]\,\ud t.
\label{eq:parent-sos}
\end{equation}
This is the noncommutative counterpart of the factorization ``Laplacian = gradient$^{\dagger}$gradient''.
Here, each $\cB_{j,t}$ can be regarded as a first-order (i.e., derivative) operator and $\cB_{j,t}[\sigma^{1/2}]=0$ for every $j,t$.
This representation supplies the structural ingredient that the Szegedy approach lacks in the quantum setting. 

After vectorization, we map the parent Hamiltonian $\cH$ to a positive semidefinite Hamiltonian $\widehat{\cH}$ on the doubled Hilbert space, whose ground state is the purified Gibbs state $|\sigma^{1/2}\rangle\!\rangle = \sum_j \sqrt{\sigma_j}\ket{\psi_j}\ket{\psi_j}$; under ergodicity this ground state is unique. Under this map, each first-order factor becomes
\begin{equation}\label{eqn:first-order-factor-vec}
\widehat{\cB}_{j,t}=I\otimes B_{j,t}^{\sharp}-(B_{j,t}^{\flat})^{\top}\otimes I,
\end{equation}
so that
\begin{equation}
\widehat{\cH}=\int_{-\infty}^{\infty} g(t)\sum_j \widehat{\cB}_{j,t}^{\dagger}\widehat{\cB}_{j,t}\,\ud t.
\label{eq:vectorized-factorization}
\end{equation}
The target state $|\sigma^{1/2}\rangle\!\rangle$ is annihilated by the whole family $\{\widehat{\cB}_{j,t}\}$. Conversely, since~\cref{eq:vectorized-factorization} expresses $\widehat{\cH}$ as a weighted sum of positive operators, every ground state of $\widehat{\cH}$ lies in the joint kernel of all $\widehat{\cB}_{j,t}$. Hence, when $\cL$ is ergodic so that $\widehat{\cH}$ has a unique ground state, this joint kernel is one-dimensional and is spanned by $|\sigma^{1/2}\rrangle$.

\subsection{Sum-of-squares factorization for DLL Gibbs sampler}
We briefly discuss the factorization of the DLL sampler due to its simplicity. For the factorization of the Davies generator and the CKG sampler, see~\Cref{append:classes-gibbs-samplers}.
With a single coupling operator, the DLL Lindbladian admits a single jump operator:
\begin{align}
    L_{\rm DLL} = \int^\infty_{-\infty} f(t) e^{iHt}Ae^{-iHt}~\d t = \sum_{\nu} q(\nu)e^{-\beta \nu/4}A_\nu,
\end{align}
where $f$ is the inverse Fourier transform of $q(\nu) e^{-\beta\nu/4}$. 
The weighting function $q(\nu) \colon \RR \to \CC$ satisfies $q(-\nu) = \overline{q(\nu)}$.
In this case, the Kossakowski matrix is of rank $J=1$: $\mathbf{C}_{\rm DLL} = vv^\dagger$, where $v$ is a column vector and $v_\nu = q(\nu)e^{-\beta\nu/4}$.
We can write the factorization of the parent Hamiltonian of $\cL_{\rm DLL}$ in the form of~\cref{eq:vectorized-factorization}, with
\begin{align}
    B^\sharp = \sum_{\nu}q(\nu) e^{-\beta \nu/4} A_{\nu},\quad B^\flat = \sum_{\nu}q(\nu) e^{\beta \nu/4} A_{\nu}.
\end{align}
Note that $B^\sharp$ coincides with the jump operator $L_{\rm DLL}$.
More generally, the rank $J$ is the number of coupling operators used in the DLL sampler. With appropriate choices of the weighting function $q$, $B^\sharp$ and $B^\flat$ can be efficiently implemented using numerical quadratures as in~\cite[Section~3.3]{ding2025efficient}.

\subsection{Quantum algorithms and complexity analysis}
We can now construct a QSVT-based quantum algorithm to prepare the purified Gibbs state $|\sigma^{1/2}\rrangle$ from \cref{eq:vectorized-factorization}. 
We truncate the time integral to $-T \le t \le T$ and discretize it on a grid $\{t_k = -T + kh\}^{N-1}_{k=0}$, where $h = 2T/N$. 
This yields a numerical approximation of the form
\begin{equation}
\widehat{\cH}\approx \BB^{\dagger}\BB,
\qquad
\BB_{j,k}=\sqrt{hg(t_k)}\,\widehat{\cB}_{j,t_k},
\label{eq:discrete-B}
\end{equation}
where $\BB$ is a rectangular operator whose smallest right singular vector is approximately $|\sr\rrangle$ up to discretization error.
A quadrature analysis shows that the error can be made at most $\epsilon$ by choosing
\begin{equation}
T=\Or\!\left(\beta\log\frac{J}{\epsilon}\right),
\quad
N=\Or\!\left(\log^2\frac{J}{\epsilon}+\beta\norm{H}\log\frac{J}{\epsilon}\right),
\label{eq:quadrature}
\end{equation}
where $J$ is the number of jump-operator channels in the sum-of-squares decomposition as in~\cref{eq:vectorized-factorization}, and $N$ is the number of  quadrature points.  
The linear dependence on $\norm{H}$ comes from the spectral width of the vectorized time-evolution generator $K = -H^\top \otimes I + I\otimes H$, see~\cref{append:quadrature-error-analysis}. 

\begin{figure}[!ht]
\centering
\makebox[\columnwidth]{\hspace*{0.22\columnwidth}\textbf{(a)}\hfill\textbf{(b)}\hspace*{0.22\columnwidth}}\par
\vspace{0.2em}
\begin{overpic}[width=\columnwidth,trim={3cm 14cm 3cm 14cm},clip]{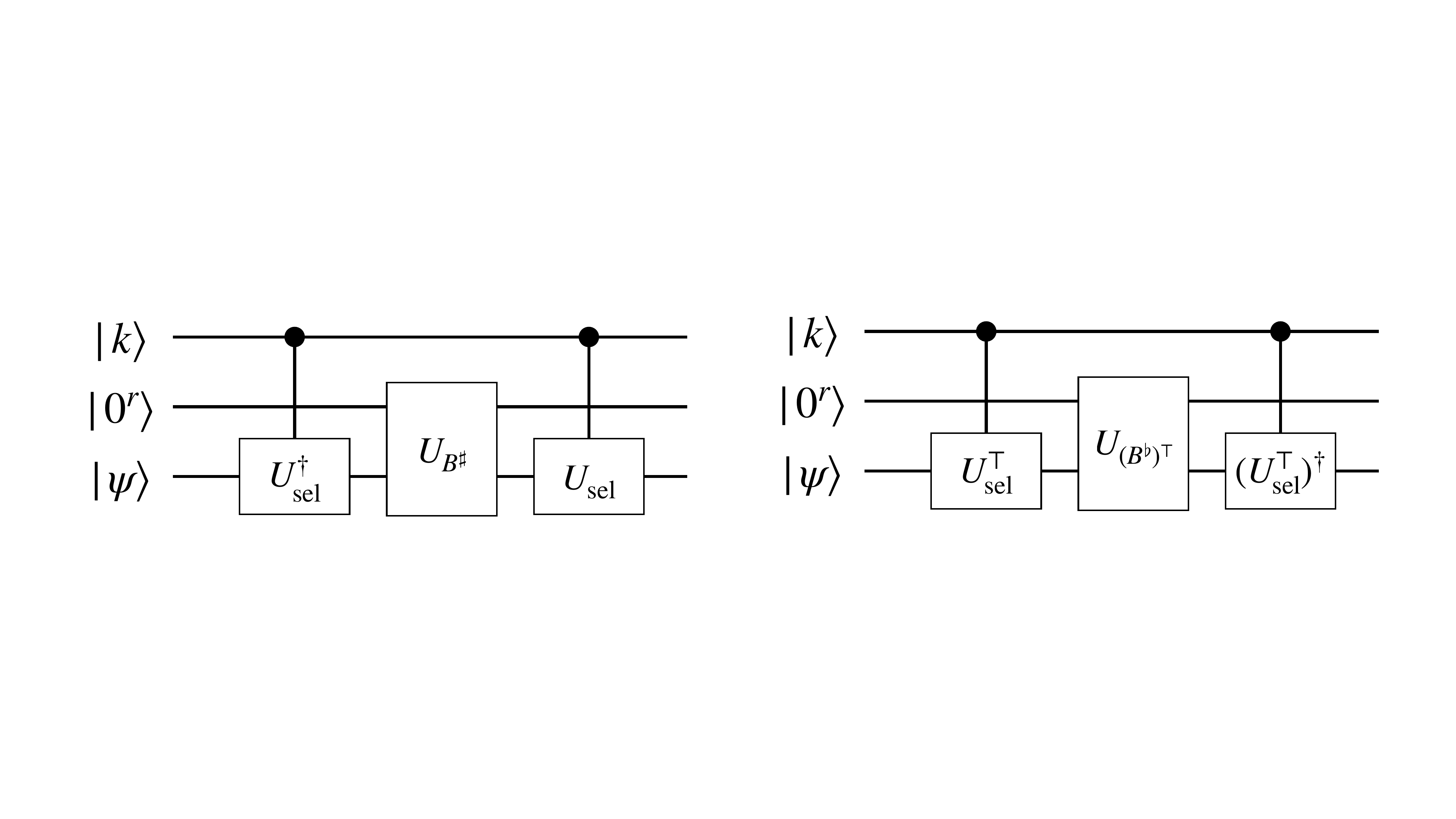}
\put(50,95){\color{white}\rule{1.2cm}{0.45cm}}
\end{overpic}

\vspace{0.8em}
\textbf{(c)}\par
\vspace{0.2em}
\begin{overpic}[width=\columnwidth,trim={8cm 12cm 8cm 9cm},clip]{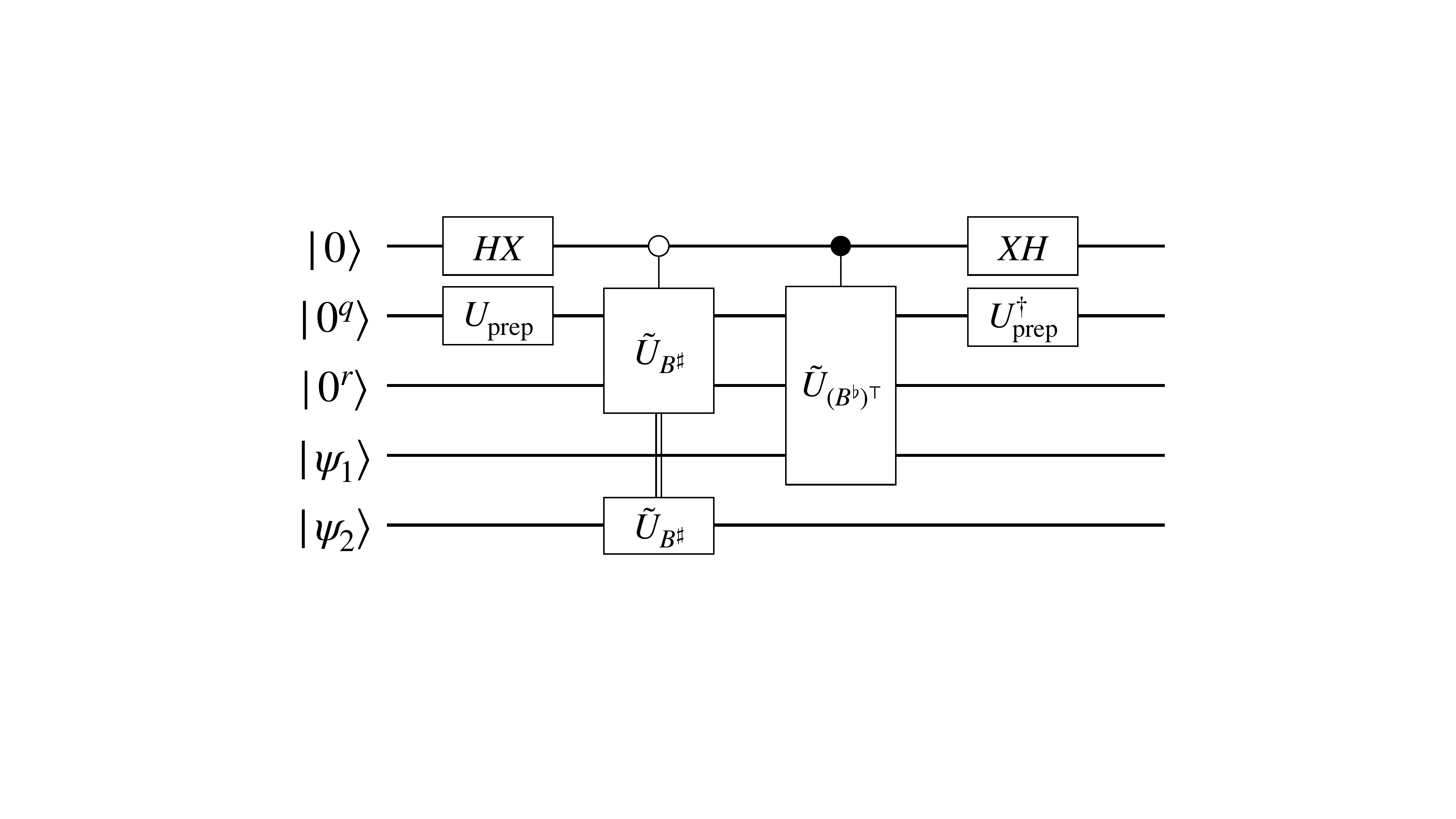}
\end{overpic}
\caption{Circuit schematics for the block-encoding of the discretized first-order operator $\BB$. 
Panels (a) and (b) show the dressed block-encodings of $\tilde{U}_{B^\sharp}$ and $\tilde{U}_{(B^\flat)^\top}$. 
Panel (c) shows the full circuit assembling these ingredients with the quadrature-weight preparation oracle $U_{\rm prep}$ and a branch ancilla to realize the linear combination corresponding to $I\otimes B^\sharp_{j,t_k}-(B^\flat_{j,t_k})^\top\otimes I$.}
\label{fig:block-encoding-B}
\end{figure}

Recall that $\widehat{\cH}$ has the same spectrum as $-\cL$. Since $\BB^{\dagger}\BB$ approximates $\widehat{\cH}$ and the nonzero eigenvalues of $\BB^{\dagger}\BB$ are the squares of the nonzero singular values of $\BB$, the smallest nonzero singular value of $\BB$ is $\sqrt{\Delta}$ up to discretization error, where $\Delta$ is the Lindbladian gap. 
Therefore, QSVT yields an approximate projector onto the ground right-singular subspace of $\BB$ using $\widetilde{\cO}(1/\sqrt{\Delta})$ queries to a block-encoding of $\BB$, thereby mapping a warm-start state to the desired purified Gibbs state with constant success probability.
The block-encoding of $\BB$ can be built from block-encodings of $B_j^\sharp$ and $B_j^\flat$, controlled time evolution under $H$, and a state-preparation oracle $U_{\rm prep}$ for the quadrature weights $g(t_k)$; see~\cref{fig:block-encoding-B}.
For efficiently implementable samplers such as CKG~\cite{chen2023efficient} and DLL~\cite{ding2025efficient}, the required block-encodings of $B_j^{\sharp}$ and $B_j^{\flat}$ can be obtained via weighted operator Fourier transforms at a cost comparable to implementing the samplers themselves; see~\Cref{append:classes-gibbs-samplers} for details.
The controlled evolution of a general local Hamiltonian $H$ for time $t$ can be simulated with gate complexity linear in $t$ under a block-encoding input model for $H$.
Since the kernel function $g(t)$ is essentially supported on an interval of length $\mathcal{O}(\beta)$, the state $\ket{g}$ can be prepared using quantum signal processing with $\poly\log(1/\varepsilon)$ gates and constant ancillas, where $\varepsilon$ is the precision parameter~\cite{mcardle2022quantum}.
For simplicity, we assume black-box access to block-encodings of $B^\sharp$ and $B^\flat$ (denoted as $U_{B^\sharp}$ and $U_{B^\flat}$) and to the state-preparation oracle $U_{\rm prep}$ that prepares $\ket{g}$.
\begin{thm}\label{thm:main-text}
Suppose that the Lindbladian $\cL$ admits a factorization with $J$ jump operators as in~\cref{eq:vectorized-factorization}, and we have access to a warm-start state $|\phi\rangle$ such that $|\langle\phi|\sigma^{1/2}\rangle\!\rangle|^2=\Omega(1)$.
Then, we can prepare an $\epsilon$-approximation of $|\sr\rrangle$ using
\begin{equation}
\Or\!\left(\sqrt{\frac{J}{\Delta}}\log\frac{1}{\epsilon}\right)
\label{eq:main-complexity}
\end{equation}
queries to $U_{B^\sharp}$, $U_{(B^\flat)^\top}$, $U_{\rm prep}$ (and their inverses), and $\Or(1)$ copies of the warm start.
The maximal Hamiltonian evolution time is $\Or(\beta\log(J/\Delta\epsilon))$.
\end{thm}

A detailed statement and proof of~\cref{thm:main-text} are given in~\cref{append:circuit-implementation-cost}.

\prlsection{Application: warm starts and slow thermal refinement}

Generating the warm-start state $\ket{\phi}$ is difficult even for quantum computers: in the low-temperature limit $\beta \to \infty$, preparing a state with constant overlap with $|\sr\rrangle$ would in general yield constant overlap with a ground-state purification of $H$, so one should not expect an efficient generic procedure.
Nevertheless, we show that in certain Hamiltonian models easy warm-start preparation can \emph{coexist} with slow thermal relaxation, even from the same warm start. For these models, our quantum algorithm achieves quadratic speedup in an end-to-end manner.

For odd $n$, we define $M=\sum^n_{i=1}Z_i$ and consider the inhomogeneous Curie--Weiss Hamiltonian
\begin{equation}\label{eq:model}
  H= H_{\rm cl} -\sum_{i<j}J_{ij}X_iX_j,\qquad H_{\rm cl} = -\frac{M^2}{2n}-\frac Mn.
\end{equation}
Here, $J_{i,j} = J_{j,i}$ and $0 < J_{i,j} \le \kappa /n$ for some sufficiently small $\kappa$. The $1/n$ normalization keeps the dense two-body interactions extensive and the quantum model retains the classical critical inverse temperature $\beta_c=1$.

The classical Hamiltonian $H_{\rm cl}$ is a ferromagnetic Curie--Weiss model. At low temperature, it exhibits two phases with opposite magnetizations, and the weak longitudinal field $-M/n$ breaks the spin-flip symmetry and favors the positive phase.
Our key observation is that this dominant phase can be accurately approximated by a product of independent spin states. Purifying this approximation gives a product of two-qubit states, each pairing a system qubit with an ancillary qubit, which can be prepared in constant depth. Since the positive phase carries most of the Gibbs weight, this state has large constant overlap with the purified Gibbs state. We show that this overlap remains large under the weak $XX$ interactions, giving an efficiently preparable warm start for the quantum model.

Nevertheless, thermal relaxation can face a serious bottleneck: the region near zero magnetization has exponentially small Gibbs weight, suppressing transfer between the two phases. For the Gibbs sampler constructed by Chen, Kastoryano, and Gily\'en~\cite{chen2023efficient} using single-site Pauli couplings, we prove that this slow-transition mechanism forces an exponentially small spectral gap.
This bottleneck persists even when the sampler starts from the physical marginal of our warm start. The negative phase carries a small equilibrium weight that does not vanish as $n$ grows. Reaching a sufficiently small fixed error therefore requires restoring this missing weight, and we prove that this refinement takes time exponential in $n$.
Details of the warm-start construction and proofs of slow mixing and thermal refinement are given in~\Cref{append:easy-warm-starts}.

\prlsection{Discussion}
We have shown that quantum Gibbs sampling admits a quadratic improvement in gap dependence without relying on the quantum-walk mechanism. For a broad class of KMS-detailed-balanced Lindbladians in Kossakowski form, the factorization $\widehat{\cH}=\BB^\dagger\BB$ turns purified Gibbs-state preparation into a singular-value problem for a noncommutative first-order operator $\BB$, leading to a $1/\sqrt{\Delta}$ dependence instead of $1/\Delta$. 
This extends the operator-level acceleration of classical continuous-space Gibbs sampling via the Witten-Laplacian factorization~\cite{leng2026operator} to the quantum setting, but with an additional time parameter $t$ in the first-order factors $\cB_{j,t}$.
We note that factorization-based gap amplification has appeared in quantum simulation, especially for (near) frustration-free Hamiltonians~\cite{low2017hamiltonian,low2025fast,somma2013spectral,king2026quantum}; however, its application to quantum Gibbs sampling requires detailed analysis of the parent Hamiltonian, which is one of the key technical contributions of this work.

Purified Gibbs states are also of independent interest in many-body physics and quantum gravity~\cite{brown2023quantum,nezami2021quantum,schuster2022many,jafferis2022traversable,maldacena2018eternal,zhu2020generation,wu2019variational}. For specific models such as the Sachdev-Ye-Kitaev (SYK) model, the main bottleneck is the efficient preparation of warm-start states. 
The inhomogeneous Curie--Weiss models introduced here provide concrete physical settings in which our algorithm achieves an end-to-end quadratic improvement in spectral-gap dependence for purified Gibbs-state preparation. Determining whether this advantage extends to other many-body systems is an interesting direction for future work. 

\section*{Acknowledgements}
This work is partially supported by the Simons Quantum Postdoctoral Fellowship (J.L., J.J.), by a Simons Investigator Award in Mathematics through Grant No. 825053 (J.L., L.L.), and by the U.S. Department of Energy, Office of Science, Accelerated Research in Quantum Computing Centers, Quantum Utility through Advanced Computational Quantum Algorithms, grant no. DE-SC0025572 (L.L.). 

\nocite{*}

\bibliographystyle{apsrev4-2}
\bibliography{refs}

\clearpage
\appendix
\renewcommand{\thesubsection}{\thesection.\arabic{subsection}}
\renewcommand{\thesubsubsection}{\thesubsection.\arabic{subsubsection}}
\crefalias{section}{appendix}
\crefalias{subsection}{appendix}
\crefalias{subsubsection}{appendix}
\Crefname{appendix}{Appendix}{Appendices}
\crefname{appendix}{appendix}{appendices}
\section{Sum-of-squares factorization of KMS DB Lindbladians}\label{append:factorization-KMS-lindbladian}

\paragraph*{Preliminaries.}
Given an $n$-qubit Hamiltonian $H$ and an inverse temperature $\beta > 0$, the Gibbs state is defined by:
\begin{align}
    \sigma = \frac{1}{Z_\beta} \exp(-\beta H), \quad Z_\beta = \Tr[\exp(-\beta H)].
\end{align}
With the Gibbs state $\sigma$, we define the \emph{modular} operator and the \emph{weighting} operator, respectively:
\begin{align}
    \mdr(X) = \sigma X \sigma^{-1},\quad \wgt(X) = \sigma^{1/2} X \sigma^{1/2}.
\end{align}

Given two matrices $A,B$, we denote the \textit{Hilbert--Schmidt} (HS) inner product as $\langle A, B\rangle = \Tr(A^\dagger B)$.
For $0 \le s \le 1$, we can define an $s$-parametrized family of inner products: 
\begin{align}
    \langle A,B\rangle_{\sigma, s} \coloneqq \Tr(A^\dagger \sigma^{1-s} B \sigma^{s}),
\end{align}
where $\langle \cdot,\cdot\rangle_{\sigma,1}$ and $\langle \cdot,\cdot\rangle_{\sigma,1/2}$ are called the \emph{Gelfand--Naimark--Segal (GNS)} and \emph{Kubo--Martin--Schwinger (KMS)} inner products, respectively.
For simplicity, we will write the KMS inner product as $\langle \cdot,\cdot\rangle_{\sigma}$ when there is no ambiguity.

A Lindblad master equation is of the form:
\begin{align}\label{eqn:lindblad}
    \dot{\rho} = \cL[\rho],\quad \cL[\cdot] = -i[G, \cdot] + \sum_{a\in \cA} L_a \cdot L^\dagger_a - \frac{1}{2}\{L^\dagger_aL_a,\cdot\}, 
\end{align}
where $\rho(t)$ for $t \ge 0$ is a time-dependent density operator, and the operator $\mathcal{L}$ is called a \emph{Lindbladian}. The coherent part $G$ is a Hermitian operator, and $\{L_a\}_{a\in\cA}$ (where $\cA$ is an index set) is a set of jump operators.

\begin{defn}
    A Lindbladian $\cL$ satisfies the KMS detailed balance condition with respect to a (full-rank) $\sigma$ if $\cL^\dagger$ is self-adjoint with respect to the KMS inner product. In other words, for any $X$, $Y$, we have 
    \[\langle X, \cL^{\dag}(Y)\rangle_{\sigma} = \langle \cL^{\dag}(X), Y\rangle_{\sigma},\]
    and equivalently, $\Gamma_\sigma \cL^{\dag} = \cL\Gamma_\sigma$.
\end{defn}

If $\cL$ satisfies KMS DBC, we have~\cite[Lemma 9]{ding2025efficient}:
\begin{equation}
    \mdr^{-1/2} L_a = L^\dagger_a,\quad \forall a \in \cA.
\end{equation}

When $\cL$ satisfies the KMS detailed balance condition, one can check that the Gibbs state $\sigma$ is a stationary point of $\cL$, i.e., $\cL[\sigma] = 0$. In this case, we can define the \emph{discriminant} (or parent Hamiltonian) of $\cL$.
\begin{align}\label{eqn:parent-hamiltonian}
    \cH(X) = - \wgt^{-1/2}\cL \wgt^{1/2}(X) = - \sigma^{-1/4} \cL \left[\sigma^{1/4}X\sigma^{1/4}\right]\sigma^{-1/4}.
\end{align}
We can readily verify that $\cH^\dagger = \cH$ due to the KMS detailed balance.
In the literature, the parent Hamiltonian is sometimes defined without the minus sign in~\eqref{eqn:parent-hamiltonian}. Our definition ensures that $\cH \succeq 0$.
The similarity transformation implies that $\cH$ has the same spectrum as $-\cL$, so $\cH[\sigma^{1/2}]=0$ and $\cH \ge 0$.

\subsection{Bohr-frequency components}
For an $n$-qubit Hamiltonian $H=\sum_{j} \lambda_j P_j$, where $\lambda_j$ and $P_j$ denote the eigenvalues and eigenprojectors of $H$, we define the set of Bohr frequencies: 
\begin{align}
    B_H = \{\nu = \lambda_i - \lambda_j \colon\, \lambda_i, \lambda_j \in \mathrm{Spec}(H)\},
\end{align}
where $\mathrm{Spec}(H)$ denotes the spectrum (i.e., the set of eigenvalues) of $H$.
For a (non-zero) matrix (i.e., coupling operator) $A \in \CC^{2^n\times 2^n}$, one can write
\begin{align}
    A = \sum_{\nu \in B_H} A_\nu,\quad A_\nu \coloneqq \sum_{\lambda_i - \lambda_j = \nu} P_i A P_j, 
\end{align}
where $P_i$ is the projection onto the $\lambda_i$-eigenspace. 
We denote 
\begin{align}
    A(t) = e^{iHt}A e^{-iHt} =\sum_{\nu \in B_H} e^{i \nu t} A_{\nu}.
\end{align}
For a function $\varphi(t)$, we define its Fourier transform 
\begin{align}
    \hat{\varphi}(\xi) := \int^\infty_{-\infty} e^{i\xi t} \varphi(t)~\d t,
\end{align}
and the inverse Fourier transform gives $\varphi(t) = \frac{1}{2\pi}\int_{-\infty}^{+\infty} e^{-i\xi t}\hat{\varphi}(\xi)~\d \xi$. 
The weighted operator Fourier transform (WOFT) implements a linear combination of the Bohr-frequency components weighted by $\hat{\varphi}$,
\begin{align}\label{eqn:woft}
    \int^\infty_{-\infty} A(t) \varphi(t)~\d t = \sum_{\nu \in B_H} \hat{\varphi}(\nu)A_\nu.
\end{align}

\subsection{Symmetric bilinear form with Kossakowski matrix}
\label{append:symmetric-form-kossakowski}

Suppose that the Lindbladian $\cL$ satisfies KMS detailed balance. 
For two matrices $X$ and $Y$, we define the symmetric bilinear form of $\cL$ as:
\begin{align}
    \cE(X,Y) = -\langle X,\cL^\dagger(Y) \rangle_{\sigma} = - \Tr[X^\dagger \sigma^{1/2} \cL^\dagger(Y) \sigma^{1/2}].
\end{align}
In particular, $\cE(X,X)$ is called the \emph{Dirichlet form} associated with (the observable matrix) $X$.
Since the Lindbladian $\cL$ has a non-positive spectrum, we have $\cE(X,X) \ge 0$ for any $X$. 

In the following, we show that any KMS-detailed-balanced Lindbladian in the ``Kossakowski form'' admits a sum-of-squares representation of the Dirichlet form.
Notably, this canonical Kossakowski form encompasses many interesting classes of quantum Gibbs samplers, including the Davies generator, CKG~\cite{chen2023efficient}, and DLL~\cite{ding2025efficient} as special cases.

\begin{thm}\label{thm:sum-of-squares-kossa}
    Suppose that a Lindbladian operator is of the form:
    \begin{align}\label{eqn:kossa-form-statement}
        \cL(\cdot) = -i[G, \cdot] + \sum_{\nu, \nu'} \alpha_{\nu,\nu'} \left(A_\nu \cdot A^\dagger_{\nu'} - \frac{1}{2}\{A^\dagger_{\nu'} A_\nu, \cdot\}\right),
    \end{align}
    where $A$ is a non-zero Hermitian matrix and $\mathbf{C} = (\alpha_{\nu,\nu'})_{\nu,\nu'\in B_H} \succeq 0$ is the Kossakowski matrix, and the coherent term is given by
    \begin{align}\label{eqn:kossa-coherent-G}
        G = \frac{1}{2i}\sum_{\nu, \nu'} \tanh\left(-\frac{\beta (\nu-\nu')}{4}\right) \alpha_{\nu, \nu'} (A_{\nu'})^\dagger A_\nu .
    \end{align}
    Suppose that we have a matrix $Q\in \CC^{|B_H|\times J}$ such that $\mathbf{C} = QQ^\dagger$, and  $\alpha_{-\nu',-\nu} = \alpha_{\nu, \nu'}e^{\beta(\nu+\nu')/2}$.
    We define $g(t) \coloneqq \frac{1}{\beta \cosh(2\pi t/\beta)}$. Then, the symmetric bilinear form of $\cL$ can be represented as follows,
    \begin{equation}\label{eqn:kossa-dirichlet-freq}
      \begin{aligned}
        \cE(X,Y) &= \int^\infty_{-\infty}g(t)  \sum^J_{j=1}\langle [B_{j,t}, X], [B_{j,t}, Y] \rangle_\sigma~\d t, \\
        B_{j,t} &\coloneqq e^{iHt}\left(\sum_\nu Q_{\nu, j} e^{\beta \nu/4}A_\nu\right)e^{-iHt} = \sum_{\nu} Q_{\nu, j} e^{i\nu t}e^{\beta \nu/4} A_{\nu}.
      \end{aligned}
    \end{equation}
\end{thm}

\begin{rem}
    For simplicity, the Lindbladian operator as in~\cref{eqn:kossa-form-statement} only involves a single coupling operator $A$. It is straightforward to generalize the result to a set of coupling operators $\{A_k\}$, and the sum-of-squares factorization still holds by summing over all $\{A_k\}$.
\end{rem}

\begin{rem}
    The factored representation of the symmetric bilinear form as in~\cref{eqn:kossa-dirichlet-freq} was previously known for the Davies generators~\cite[Proposition 2.9]{basso2025quantum} and the CKG sampler, see~\cite[Lemma X.2]{chen2025quantum} and~\cite[Lemma C.1]{rouze2025efficient}.
    Recently, a more general form has been reported in~\cite[Appendix A]{slezak2026polynomial}, but only in the Bohr-frequency resolved form therefore does not immediately lead to efficient quantum implementation.
\end{rem}

\begin{proof}
    First, by the identity 
    \[1 \pm \tanh(x) = \frac{e^{\pm x}}{\cosh(x)},\]
    we have that
    \begin{align}
        \cL^\dagger(Y) & = \sum_{\nu,\nu'} \alpha_{\nu,\nu'}^* \left(
        A_{\nu}^\dagger Y A_{\nu'}- \frac{1}{2} \left( \frac{e^{\beta(\nu'-\nu)/4}}{\cosh(\beta(\nu'-\nu)/4)}A_{\nu}^\dagger A_{\nu'} Y +  \frac{e^{-\beta(\nu'-\nu)/4}}{\cosh(\beta(\nu'-\nu)/4)}\,Y A_{\nu}^\dagger A_{\nu'}  \right)\right)\\    
        &= \sum_{\nu,\nu'} \alpha_{\nu,\nu'}\left(A^\dagger_{\nu'} Y A_\nu - \frac{1}{2}\left(\frac{e^{\beta(\nu-\nu')/4}}{\cosh(\beta(\nu-\nu')/4)}A^\dagger_{\nu'} A_\nu Y +  \frac{e^{\beta(\nu'-\nu)/4}}{\cosh(\beta(\nu-\nu')/4)} YA^\dagger_{\nu'} A_\nu \right)\right)  \label{eq:B11} \\
        &= \sum_{\nu,\nu'} \alpha_{\nu,\nu'} \left( \frac{e^{\beta(\nu-\nu')/4}}{2\cosh(\beta(\nu-\nu')/4)} A^\dagger_{\nu'}[Y,A_\nu] + \frac{e^{\beta(\nu'-\nu)/4}}{2\cosh(\beta(\nu-\nu')/4)} [A_{\nu'}^\dagger, Y]A_{\nu} \right)\label{eq:B12}\\
        &= \sum_{\nu,\nu'} \alpha_{\nu,\nu'} \left( \frac{e^{\beta(\nu-\nu')/4}}{2\cosh(\beta(\nu-\nu')/4)} A^\dagger_{\nu'}[Y,A_\nu] + \frac{e^{\beta(\nu'-\nu)/4}}{2\cosh(\beta(\nu-\nu')/4)} e^{\beta(\nu+\nu')/2} [A_\nu, Y]A^\dagger_{\nu'} \right)
    \end{align}
    In Eq.~(\ref{eq:B11}), we use the assumption that the Kossakowski matrix $\mathbf{C}$ is Hermitian, which implies $\alpha_{\nu,\nu'}^* = \alpha_{\nu',\nu}$, and then relabel the indices $\nu, \nu'$ as $\nu', \nu$. 
    To obtain the last identity, we relabel $\nu \mapsto -\nu'$ and $\nu' \mapsto -\nu$ in the second term in~\cref{eq:B12} and use the fact that 
    \[\alpha_{-\nu',-\nu} = \alpha_{\nu, \nu'}e^{\beta(\nu+\nu')/2},\quad A_{-\nu} = A^\dagger_\nu.\]
    Then, direct calculation shows that 
    \begin{align}\label{eqn:A20-inner-1}
        \langle X, A^\dagger_{\nu'} [Y, A_{\nu}]\rangle_\sigma 
        &= tr\left(X^\dagger \sr A^\dagger_{\nu'} [Y, A_{\nu}] \sr\right)= e^{\beta\nu'/2} \langle A_{\nu'}X, [Y, A_{\nu}]\rangle_\sigma,
    \end{align}
    \begin{align}\label{eqn:A20-inner-2}
        \langle X, [A_{\nu}, Y]A^\dagger_{\nu'}\rangle_\sigma 
        &= tr\left(X^\dagger \sr [A_{\nu}, Y]A^\dagger_{\nu'} \sr\right) = - e^{-\beta\nu'/2} \langle X A_{\nu'}, [Y, A_{\nu}]\rangle_\sigma.
    \end{align}
    Plugging~\cref{eqn:A20-inner-1} and~\cref{eqn:A20-inner-2} into the symmetric bilinear form $\cE(X,Y)$, we obtain
    \begin{align}\label{eqn:cE-1}
        \cE(X,Y) = \sum_{\nu,\nu'} \alpha_{\nu,\nu'}\frac{e^{\beta(\nu+\nu')/4}}{2\cosh(\beta(\nu-\nu')/4)} \langle [A_{\nu'}, X], [A_\nu, Y] \rangle_\sigma.
    \end{align}
    Finally, by plugging the identity (note that $\cosh(x)$ is symmetric),
    \begin{align}\label{eqn:cosh-fourier}
        \frac{1}{2\cosh\left(\beta(\nu - \nu')/4 \right)} = \int_{-\infty}^{+\infty} g(t) e^{i(\nu-\nu')t} \d t,
    \end{align}
    into~\cref{eqn:cE-1} and using the factorization of the Kossakowski matrix $\mathbf{C} = QQ^\dagger$, we can express the bilinear form as 
    \begin{align}
        \cE(X,Y) &= \int_{-\infty}^{+\infty} g(t)\sum_{\nu,\nu'} \alpha_{\nu,\nu'}\langle [e^{i\nu' t}e^{\beta \nu'/4}A_{\nu'}, X], [e^{i\nu t} e^{\beta \nu/4}A_\nu, Y] \rangle_\sigma  \d t\\
        &= \int_{-\infty}^{+\infty} g(t) \sum_{\nu,\nu'} \sum_{j} Q_{\nu,j}Q^\dagger_{j,\nu'} \langle [e^{i\nu' t}e^{\beta \nu'/4}A_{\nu'}, X], [e^{i\nu t} e^{\beta \nu/4}A_\nu, Y] \rangle_\sigma \d t\\
        &= \int_{-\infty}^{+\infty} g(t) \sum_{j} \langle [\sum_{\nu'}Q_{\nu',j}e^{i\nu' t}e^{\beta \nu'/4}A_{\nu'}, X], [\sum_\nu Q_{\nu,j} e^{i\nu t} e^{\beta \nu/4}A_\nu, Y] \rangle_\sigma \d t.
    \end{align}
    This proves the theorem.
\end{proof}
In what follows, to simplify notation, we also write $B_{j,t} = e^{iHt}B_je^{-iHt}$, where $B_j$ denotes the linear combination of $A_\nu$ according to the Kossakowski factorization:
\begin{align}
    \label{eq:Bv}
    B_j \coloneqq \sum_\nu Q_{\nu, j} e^{\beta \nu/4}A_\nu.
\end{align}

\subsection{Factorization of the parent Hamiltonian}

Recall that the parent Hamiltonian $\cH$ is defined as 
\begin{equation}
    \cH[\cdot] = -\sigma^{-1/4}\cL[\sigma^{1/4}\cdot\sigma^{1/4}]\sigma^{-1/4} = - \Gamma_{\sigma}^{-\frac{1}{2}} \cL \Gamma_{\sigma}^{\frac{1}{2}}[\cdot], 
\end{equation}
When $\cL$ satisfies KMS DBC, we know that $\cH$ is Hermitian w.r.t. Hilbert-Schmidt inner product, thus
\begin{align}\label{eqn:hermiticity-parent-hamiltonian}
    \cH = \cH^\dagger = - \Gamma_{\sigma}^{\frac{1}{2}} \cL^\dagger \Gamma_{\sigma}^{-\frac{1}{2}}.
\end{align}

The sum-of-squares representation of Lindbladian bilinear form (\cref{thm:sum-of-squares-kossa}) immediately leads to a factorization of the parent Hamiltonian, as stated in the following result. 

\begin{prop}\label{prop:parent-sum-of-squares}
    Suppose that the Lindbladian $\cL$ satisfies the KMS detailed balance condition and its symmetric bilinear form admits a factorization as in~\cref{eqn:kossa-dirichlet-freq}.
    Then, the parent Hamiltonian of $\cL$ can be represented as 
    \begin{align} 
        \cH &= \int^\infty_{-\infty} g(t) \sum^J_{j=1} \left(\cB_{j,t}\right)^\dagger\cB_{j,t},\quad \cB_{j,t}[X] \coloneqq B^\sharp_{j,t} X - X B^\flat_{j,t},
    \end{align}
    where $g(t) \coloneqq \frac{1}{\beta \cosh(2\pi t/\beta)}$, $B^*_{j,t} \coloneqq e^{iHt} B^{*}_j e^{-iHt}$ (for $* \in \{\sharp,\flat\}$), and
    \begin{align}
        B^\sharp_j = \mdr^{1/4}(B_j),\quad B^\flat_j = \mdr^{-1/4}(B_j).
    \end{align}
    for $B_j$ defined in Eq.~(\ref{eq:Bv}).
\end{prop}
\begin{proof}
    By the definition of bilinear form, we have
    \begin{align}
        \cE(X,Y) &= -tr\left(X^\dagger \sr \cL^\dagger(Y) \sr \right)\\ 
        &= -\langle X, (\wgt \cL^\dagger)(Y) \rangle\\
        &= \langle X, \wgt^{1/2}(- \wgt^{1/2}\cL^\dagger \wgt^{-1/2})[\wgt^{1/2}(Y)]\rangle \label{eqn:A34-step-3}\\
        & = \langle X, \wgt^{1/2} \cH \wgt^{1/2}(Y) \rangle,
    \end{align}
    where in~\cref{eqn:A34-step-3}, we use the Hermiticity of the parent Hamiltonian~\cref{eqn:hermiticity-parent-hamiltonian}.
    Meanwhile, by~\cref{thm:sum-of-squares-kossa}, we have (where $\cC_{j,t}(X) \coloneqq [B_j(t), X]$ denotes the commutator in~\cref{eqn:kossa-dirichlet-freq})
    \begin{align}
        \cE(X,Y) &= \int^\infty_{-\infty} g(t)\sum_j \langle \cC_{t,j}(X), \wgt\cC_{t,j}(Y) \rangle  \d t
        = \left\langle X,  \int^\infty_{-\infty} g(t)\sum^J_{j=1} (\cC_{t,j})^\dagger \wgt \cC_{t,j}(Y)\d t \right\rangle.
    \end{align}
    By identifying the two equations, we have
    \begin{align}\label{eqn:discriminant-factor}
        \cH &= \wgt^{-1/2}\left(\int^\infty_{-\infty} g(t)\sum^J_{j=1} (\cC_{t,j})^\dagger \wgt \cC_{t,j}\d t\right) \wgt^{-1/2} 
        = \int^\infty_{-\infty} g(t)\sum_j (\cB_{j,t})^\dagger\cB_{j,t}~\d t,
    \end{align}
    where
    \begin{align}
        \cB_{j,t}[X] &= \wgt^{1/2} \circ \cC_{j,t}\circ \wgt^{-1/2}[X]\\
        &= \sigma^{1/4}B_{j,t}\sigma^{-1/4}X - X\sigma^{-1/4}B_{j,t}\sigma^{1/4}\\
        &= B^\sharp_{j,t} X - X B^\flat_{j,t}.
    \end{align}
\end{proof}

\subsection{Several classes of quantum Gibbs samplers}\label{append:classes-gibbs-samplers}

As we mentioned, the Kossakowski form in~\cref{thm:sum-of-squares-kossa} encompasses several major classes of known quantum Gibbs samplers.
In this subsection, we give explicit calculations of the factorization of the parent Hamiltonians for the Davies generator and two prominent classes of KMS detailed balanced Lindbladians. 
For simplicity, we only consider a single jump operator in the Lindbladian; the generalization to multiple jumps should be straightforward.

\subsubsection{Davies generators}
The Davies generator is described by the following Lindbladian:
\begin{align}\label{eqn:davies}
    \mathcal{L}_\beta(\rho) = \sum_{\nu \in B_H} \gamma(\nu) \left(A_\nu \rho A^\dagger_\nu - \frac{1}{2}\{A^\dagger_\nu A_\nu, \rho\}\right),
\end{align}
where $\gamma(\nu) = \min(1, e^{-\beta \nu})$. This Lindbladian satisfies KMS detailed balance condition and its Kossakowski matrix is diagonal: $\alpha_{\nu,\nu} = \gamma(\nu)$.
Therefore, we can write $\mathbf{C}_{\rm Davies} = QQ^\dagger$, where $Q\in \mathbb{C}^{|B_H|\times |B_H|}$ and 
\begin{align}
    B_\nu(t) = \sqrt{\gamma(\nu)} e^{i\nu t}e^{\beta \nu/4} A_\nu = e^{iHt} \left(\sqrt{\gamma(\nu)}e^{\beta \nu/4} A_\nu\right) e^{-iHt}.
\end{align}
Now, by~\cref{prop:parent-sum-of-squares}, we can compute the factors in the parent Hamiltonian (ignoring the phase $e^{i\nu t}$):
\begin{align}
    B^\sharp_\nu &= \sigma^{1/4}B_\nu\sigma^{-1/4} = \sqrt{\gamma(\nu)} A_\nu = \gamma^+(\nu) A_\nu,\\
    B^\flat_\nu &= \sigma^{-1/4}B_\nu\sigma^{1/4} = \sqrt{\gamma(\nu)}e^{\beta \nu/2} A_\nu = \gamma^-(\nu) A_\nu,
\end{align}
where  
\begin{align}
    \gamma^-(\nu) = \min(e^{\beta \nu/2},1),\quad \gamma^+(\nu) = \min(1, e^{-\beta \nu/2}).
\end{align}

\subsubsection{The Chen--Kastoryano--Gily\'en sampler}  
The Gibbs sampler constructed by Chen, Kastoryano, and Gily\'en (CKG,~\cite{chen2023efficient}) satisfies exact KMS detailed balance. The Lindbladian is parametrized by three parameters $\omega_\gamma,\sigma_\gamma$ and $\sigma_E$, as given by 
\begin{equation}
    \cL(\rho) = -i[G,\rho] + \int^\infty_{-\infty} \gamma(\omega) \left( \widehat{A}(\omega) \rho (\widehat{A}(\omega))^\dagger - \frac{1}{2}\left\{(\widehat{A}(\omega))^\dagger \widehat{A}(\omega), \rho\right\}\right),
\end{equation}
where $G$ is the coherent part (details omitted here, as they are completely determined by the dissipative part due to KMS detailed balance), and
\begin{align}
    \gamma(\omega) &\coloneqq\exp\left(-\frac{(\omega+w_\gamma)^2}{2\sigma_\gamma^2}\right) \text{ with shift $\sigma_\gamma^2:=\frac{2 w_\gamma}{\beta}-\sigma_E^2$},\\
    \widehat{A}(\omega) &= \frac{1}{\sqrt{2\pi}}\int^\infty_{-\infty} e^{iHt} A e^{-iHt} e^{-i\omega t} f(t)~\d t, \text{ where } f(t) = e^{-\sigma^2_E t^2}\sqrt{\sigma_E \sqrt{2/\pi}}.
\end{align}
We can write the CKG Lindbladian in the canonical form~\cref{eqn:kossa-form-statement}, where the Kossakowski matrix $\mathbf{C}_{\rm CKG} = (\alpha_{\nu,\nu'})_{\nu, \nu' \in B_H}$ is given by
\begin{align}\label{eq:alphav1v2}
    \alpha_{\nu,\nu'} \coloneqq \int_{-\infty}^{+\infty} \gamma(\omega) \beta_{\nu}^\omega \beta_{\nu'}^\omega  \d \omega, \text{ where } \beta_\nu^\omega &\coloneqq \frac{1}{\sqrt{2\sigma_E\sqrt{2\pi}}}\exp\left(-\frac{(\omega-\nu)^2}{4\sigma_E^2}\right).
\end{align}
For simplicity, we choose $\sigma_E = \sigma_\gamma = \omega_\gamma = 1/\beta$. 
Note that the Kossakowski matrix can be written as 
\begin{align}
    \mathbf{C}_{\rm CKG} = \int \gamma(\omega) vv^\dagger d \omega, \text{ where $v$ is a column vector},  v_\nu = \beta^\omega_\nu.
\end{align}
By~\cref{thm:sum-of-squares-kossa}, the Dirichlet form can be written as 
\begin{align}
    \cE(X,Y) &= \int^\infty_{-\infty} \int^\infty_{-\infty} g(t) \gamma(\omega) \langle [B_{t,\omega}, X], [B_{t,\omega},Y] \rangle_\sigma~\d t~\d \omega,\\
    \text{where~} B_{t,\omega} &\coloneqq \sum_{\nu} \beta^\omega_\nu e^{i\nu t}e^{\beta \nu/4} A_{\nu} = e^{iHt}\left(\sum_{\nu}\beta^\omega_\nu e^{\beta \nu/4} A_{\nu}\right)e^{-iHt}.
\end{align}
We write $B_\omega=\sum_{\nu}\beta^\omega_\nu e^{\beta \nu/4}A_{\nu}.$
Then, the factors in the parent Hamiltonian can be computed using~\cref{prop:parent-sum-of-squares},
\begin{align}
    B^\sharp_\omega &= \sigma^{1/4} B_\omega \sigma^{-1/4} = \sum_{\nu}\beta^\omega_\nu A_{\nu},\\
    B^\flat_\omega &= \sigma^{-1/4} B_\omega \sigma^{1/4} = \sum_{\nu}\beta^\omega_\nu e^{\beta \nu/2} A_{\nu}.
\end{align}
We can represent $B^\sharp_\omega$ and $B^\flat_\omega$ using the weighted operator Fourier transform:
\begin{align}
    &B^\sharp_\omega = \int^\infty_{-\infty} A(t) \varphi^{\omega}_{\sharp}(t)\d t, \quad \varphi^{\omega}_{\sharp}(t) := \frac{\sqrt{\sigma_E}}{(2\pi)^{3/4}} e^{-\sigma^2_E t^2 - i \omega t},\label{eqn:ckg-bsharp}\\
    &B^\flat_\omega = \int^\infty_{-\infty} A(t) \varphi^{\omega}_{\flat}(t)\d t, \quad \varphi^{\omega}_{\flat}(t) := \frac{\sqrt{\sigma_E}\exp(\frac{\sigma^2_E\beta^2 + 2\omega \beta}{4})}{(2\pi)^{3/4}} e^{-\sigma^2_E t^2 - i (\omega+\sigma^2_E\beta) t},\label{eqn:ckg-bflat}
\end{align}
and it follows that 
\begin{align}
    \|B^\sharp_\omega\| \le \left(\frac{\pi}{2}\right)^{1/4}\|A\|,\quad \|B^\flat_\omega\| \le \left(\frac{\pi}{2}\right)^{1/4}\exp\left(\frac{\sigma^2_E\beta^2 + 2\omega \beta}{4}\right)\|A\|.
\end{align}
For the parameter choice $\sigma_E = \sigma_\gamma = \omega_\gamma = 1/\beta$ and frequencies $\omega$ in the effective support of $\gamma(\omega)$, both norms are $\cO(\|A\|)$. 
Moreover, the representation~\cref{eqn:ckg-bsharp} and~\cref{eqn:ckg-bflat} immediately gives a recipe to efficiently implement the matrices $B^\sharp_{\omega}$ and $B^\flat_\omega$ using the weighted operator Fourier transform; for details, see~\cite[Theorem~I.2]{chen2023efficient}.

\subsubsection{The Ding--Li--Lin sampler}
Ding, Li, and Lin~\cite{ding2025efficient} introduces another class of exact KMS-detailed-balanced Lindbladian that only uses finitely many jump operators. 
Without loss of generality, we use a single coupling operator $A$, and the Lindbladian takes the following form:
\begin{align}
    \mathcal{L}(\cdot) = -i[G, \cdot] + L \cdot L^\dagger - \frac{1}{2}\{L^\dagger L, \cdot\},
\end{align}
where the jump operator 
\begin{align}
    &L = \int^\infty_{-\infty} f(t) e^{iHt}Ae^{-iHt}~\d t = \sum_{\nu \in B_H} \hat{f}(\nu)A_\nu = \sum_{\nu \in B_H} q(\nu)e^{-\beta \nu/4}A_\nu,\\
    &\text{with~~}f(t) = \frac{1}{2\pi}\int^\infty_{-\infty} q(\nu) e^{-\beta \nu/4} e^{-it\nu}~\d \nu.\label{eqn:dll-filter-f}
\end{align}
The choice of the weighting function $q(\nu) \colon \RR \to \CC$ must satisfy $q(-\nu) = \overline{q(\nu)}$.
Since we only use a single jump operator, the Kossakowski matrix is of rank 1: $\mathbf{C}_{\rm DLL} = vv^\dagger$, where $v$ is a column vector and $v_\nu = q(\nu)e^{-\beta\nu/4}$.
The coherent part $G$ is given in~\cite[Eq. (2.33)]{ding2025efficient}, which is precisely the same as in~\cref{thm:sum-of-squares-kossa}.
Therefore, the Dirichlet form can be written as a sum of squares, with  
\begin{align}
    B = \sum_\nu v_\nu e^{\beta \nu/4}A_\nu = \sum_\nu q(\nu) A_\nu
\end{align}
To compute the decomposition of the parent Hamiltonian, we use~\cref{prop:parent-sum-of-squares}:
\begin{align}
    B^\sharp &= \sigma^{1/4} B \sigma^{-1/4} = \sum_{\nu}q(\nu) e^{-\beta \nu/4} A_{\nu},\\
    B^\flat &= \sigma^{-1/4} B \sigma^{1/4} = \sum_{\nu}q(\nu) e^{\beta \nu/4} A_{\nu}.
\end{align}
We immediately observed that $B^\sharp = L$ and $B^\flat = \mdr^{-1/2}(L)$. 
Specifically, the jump operator $L = B^\sharp$ can be implemented using the weighted operator Fourier transform with a filter function $f(t)$ as in~\cref{eqn:dll-filter-f}. The operator $B^\flat$ can be implemented by
\begin{align}
 B^\flat = \int^\infty_{-\infty} e^{iHt} A e^{-iHt} \tilde{f}(t)~\d t,\quad \tilde{f}(t) = \frac{1}{2\pi}\int^\infty_{-\infty} q(\nu) e^{\beta \nu/4} e^{-it\nu}~\d \nu.
\end{align}
In~\cite[Section 3.2]{ding2025efficient}, two specific classes of the weighting functions $q(\nu)$ are discussed (namely, the Metropolis type and the Gaussian type). In these choices, $q(\nu)$ is real symmetric, which implies that $\tilde{f}(t)=f(-t)$. Moreover, the $L^1$-norm of $f$ (therefore also for $\tilde{f}$) is bounded by $\cO(1)$ given appropriate parametrizations of the weighting function $q(\nu)$, as shown in~\cite[Lemma 29]{ding2025efficient}.
Therefore, the sub-normalization factors of the matrices $B^\sharp$ and $B^\flat$ only depend on $\|A\|$, and we can apply the same quadrature method as in~\cite[Section~3.3]{ding2025efficient} to implement the block-encoded matrices.

\subsection{On the role of the weighting function $g$}

In~\cref{prop:parent-sum-of-squares}, the parent Hamiltonian $\cH$ is represented as an integral over the time domain with a weighting function $g(t) = \frac{1}{\beta \cosh(2\pi t/\beta)}$. 
This weighting results from the coherent term $G$ (see~\cref{eqn:kossa-coherent-G}). For KMS detailed balanced Lindbladian $\cL$, due to a rigid structural theorem~\cite[Theorem 10]{ding2025efficient}, this choice of the coherent term is unique once the dissipative part is fixed.
In other words, this particular choice of the weighting function $g(t)$ is a consequence of KMS detailed balance.

Since $\cB_{j,t}[\sigma^{1/2}] = 0$ for all $j$ and $t$, the parent Hamiltonian is ``frustration-free''. Therefore, with any non-negative weighting function $w(t)\ge 0$, we can define a \emph{generalized discriminant} 
\begin{equation}
    \cH_w := \int^\infty_{-\infty} w(t) \sum^J_{j=1} \left(\cB_{j,t}\right)^\dagger\cB_{j,t} ~\d t,
\end{equation}
which still admits $\sigma^{1/2}$ as a ground state. However, it is not clear if this new Hamiltonian $\cH_w$ exhibits a spectral gap that is comparable to the parent Hamiltonian $\cH$.
A systematic investigation of gap-preserving weighting functions is out of the scope of the current work. 
Here, we numerically test a natural choice 
\begin{align}
    w(t) = \frac{1}{2}\delta_0(t),
\end{align}
which can be viewed as the limit of $g(t)$ in the regime $\beta \to 0$. 
Note that $\int^\infty_{-\infty}g(t)\d t = \frac{1}{2}$ for all $\beta > 0$, so we add a $1/2$ normalization factor in the delta distribution to maintain the same probability mass.
With this delta-mass weigting function, the Hamiltonian reduces to 
\begin{equation}
    \cH_0 = \frac{1}{2} \sum^J_{j=1} \left(\cB_{j}\right)^\dagger\cB_{j}, \quad \cB_j[X] := \cB_{j,0}[X] = B^\sharp_j X - X B^\flat_j.
\end{equation}
In~\cref{fig:spectral_gap_ratio}, we numerically calculate the spectral gap of $\cH_0$ and $\cH$ using a 2-qubit TFIM model 
\begin{align}
    H = -Z_1Z_2 -0.1(X_1+X_2) 
\end{align}
and plot their ratio as a function of $\beta$. 
For small $\beta$, these two gaps are very close because the weighting function $g(t)$ is approximately a delta function (up to a 1/2 factor). As $\beta$ increase, the spectral gap of $\mathcal{H}_0$ becomes smaller than that of the parent Hamiltonian.
Interestingly, the ratio is not a monotone function in $\beta$: it first decreases to around $0.93$ as $\beta$ approaches to $1.8$, and then gradually increases and plateaus at around $0.99$.
For this model, using the factorization of $\cH_0$ instead of the full $\cH$ in the QSVT algorithms to prepare purified Gibbs state appears to be a practical approach, as their spectral gaps do not differ significantly (especially for large $\beta$). It is currently unclear if this relation always holds for larger systems or in a different parameter regime. We leave this study to future exploration.

\begin{figure}
    \centering
    \includegraphics[width=0.8\linewidth]{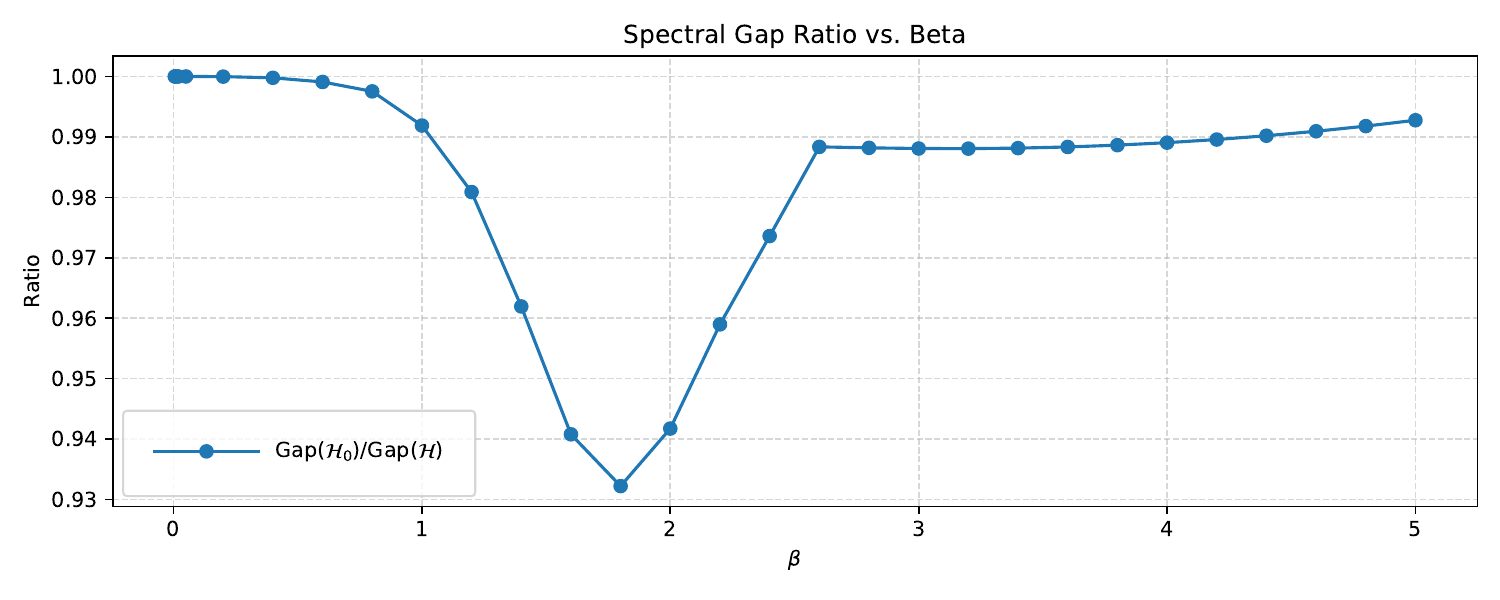}
    \caption{The ratio between the spectral gap of the new Hamiltonian $\cH_0$ and the parent Hamiltonian $\cH$.}
    \label{fig:spectral_gap_ratio}
\end{figure}

\section{Quantum algorithms for preparing purified Gibbs states}\label{append:algorithm-analysis}
\subsection{Vectorization}

Let $\mathscr{H}$ be a Hilbert space, and $\mathscr{B}(\mathscr{H})$ is the operator space over $\mathscr{H}$. Vectorization is an isomorphism from the operator space $\mathscr{B}(\mathscr{H})$ to the doubled Hilbert space $\mathscr{H}\otimes \mathscr{H}$. 
In this paper, we use the following vectorization map following the \emph{column-major} order:
\begin{itemize}
	\item The vectorization of a matrix $X=\sum_{ij} x_{ij}\ketbra{i}{j}$ is 
	\[X \mapsto |X\rrangle = \sum_{ij} x_{ij}\ket{j}\ket{i}.\]
	\item The vectorization of a super-operator $\Phi(\cdot)= \sum_{ijkl} y_{ijkl} \ketbra{i}{j}(\cdot)\ketbra{k}{l}$ is
    $$\Phi \mapsto \widehat{\Phi}=\sum_{ijkl} y_{ijkl} \ketbra{l}{k} \otimes \ketbra{i}{j}.$$
    In other words, if $\Phi(\cdot) = \sum_{ab} y_{ab}A_a\cdot B_b$, we have $\widehat{\Phi} = \sum_{ab} y_{ab} B^\top_b \otimes A_a$ where $B^\top_b$ refers to the transpose of  $B_b$.
\end{itemize}
It can be readily verified that: (1) $\widehat{\Phi} |X\rrangle=|\Phi(X)\rrangle$, (2) $\llangle X|Y\rrangle = \Tr(X^\dagger Y)$,
(3) $\widehat{\Phi}\widehat{\Phi'} = \widehat{\Phi\circ \Phi'}$, (4) $\widehat{\Phi^\dagger}= (\widehat{\Phi})^\dagger$.

We assume that the Lindbladian satisfies KMS detailed balance and is ergodic, meaning that the Gibbs state $\sigma$ is the unique stationary state of $\cL$. 
As we have shown in~\cref{append:factorization-KMS-lindbladian}, the parent Hamiltonian $\cH$ is PSD and admits a unique ground state $\sr$. 
We write the thermal state as $\sigma = \sum_j \sigma_j \ketbra{\psi_j}{\psi_j}$, where $(E_j,\ket{\psi_j})_j$ are the eigen-pairs of the Hamiltonian $H$, and $\sigma_j = e^{-\beta E_j}/Z_\beta$.
By vectorizing $\cH$, we will obtain a Hamiltonian $\widehat{\cH}$ in the doubled Hilbert space with a unique ground state 
\begin{align}
    |\sr\rrangle = \sum_j \sqrt{\sigma_j}\ket{\psi_j}\ket{\psi_j}.
\end{align}
This state is usually known as \emph{thermofield double state} or \emph{double thermal state}.

Recall from~\cref{prop:parent-sum-of-squares} that $\cH[\cdot] = \int g(t)\sum_j (\cB_{j,t})^\dagger\cB_{j,t}[\cdot]\d t$, we can also write down the vectorization of the factors of the parent Hamiltonian:
It turns out that 
\begin{align}
    \widehat{\cH} = \int^\infty_{-\infty} g(t) \sum^J_{j=1} (\widehat{\cB_{j,t}})^\dagger \widehat{\cB_{j,t}} ~\d t, \text{ where }\widehat{\cB_{j,t}} &= I\otimes B^\sharp_{j,t} - (B^\flat_{j,t})^\top \otimes I.
\end{align}

\subsection{Quadrature discretization and numerical error}\label{append:quadrature-error-analysis}

For a positive $T > 0$ and an integer $N$, we form a numerical quadrature: $\{t_k = -T + \frac{2T k}{N}\colon 0 \le k \le N-1\}$.
Then, the vectorized parent Hamiltonian, denoted here by $\widehat{H}:=\widehat{\cH}$, can be approximated by (denote $h = 2T/N$): 
\begin{align}
    \widehat{H} \approx \sum^{N-1}_{k=0}\sum^J_{j=1} g(t_k) (\widehat{B_{j,t_k}})^\dagger \widehat{B_{j,t_k}} \Delta t = \BB^\dagger \BB,
\end{align} 
where $\BB$ is a tall block matrix that stacks all the factors of the parent Hamiltonian: 
\begin{align}
    \BB^\top = \sqrt{h}\times [\sqrt{g(t_0)} \widehat{B_{1,t_0}},\dots ,\sqrt{g(t_{N-1})} \widehat{B_{1,t_{N-1}}},\sqrt{g(t_0)} \widehat{B_{2,t_0}},\dots, \sqrt{g(t_{N-1})} \widehat{B_{J,t_{N-1}}}]^\top.
\end{align}

Due to this factorization, the double thermal state $|\sr\rrangle$ is (approximately) the ground (right) singular vector of $\BB$. 
Our goal is to create a block-encoding of $\BB$ and then apply the singular value thresholding algorithm (similar to~\cite{leng2026operator}) to prepare $|\sr\rrangle$. 

\begin{lem}[Quadrature error]\label{lem:quadrature_error}
    Assuming that $\|B^*_j\| \le 1$ for all $0\le j \le J-1$ and $*\in \{\sharp,\flat\}$.
    For an $0 < \epsilon < 1$, by choosing 
    \begin{align*}
        T = \mathcal{O}\left(\beta\log \frac{J}{\epsilon} \right),\quad N =  \cO\left(\log^2\frac{J}{\epsilon} + \beta \|H\|\log\frac{J}{\epsilon} \right),
    \end{align*}
    we have $\|\widehat{H} - \BB^\dagger \BB\| \le \epsilon$.
\end{lem}

\begin{rem}
    For efficiently implementable quantum Gibbs samplers that satisfy exact KMS detailed balance (CKG and DLL), the jump operators, and more generally the operators $B^\sharp$ and $B^\flat$, can be block-encoded using the weighted operator Fourier transform. The implementation cost is mild, and the subnormalization factor mainly depends on the choice of coupling operators. Detailed discussions can be found in~\cref{append:classes-gibbs-samplers}.
    For simplicity, in the quadrature analysis we assume that the relevant individual block-encodings have been rescaled to unit subnormalization.
\end{rem}

\begin{proof}
    We first deal with the case with a single jump operator, i.e., $J=1$.
    In this case we adopt a simplified notation and write $\widehat{B}_{t_k} := \widehat{B}_{j,t_k}$.
    Recall that for $t \in \RR$, we have $B^*_{t} = e^{iHt} B^* e^{-iHt}$ for $* \in \{\sharp,\flat\}$, and
    \begin{align} 
       \widehat{B_{t}} &= I \otimes B^\sharp_{t} - \left(B^\flat_{t}\right)^\top \otimes I\\
       &= e^{-i H^\top t} \otimes e^{i Ht}\left(I\otimes B^\sharp - (B^\flat)^\top \otimes I\right) e^{i H^\top t} \otimes e^{-i Ht}
    \end{align}
    We denote $W_t := e^{-i H^\top t} \otimes e^{i Ht}$, and with the slightly abused notation $W_k := W_{t_k}$, we can write
    \begin{align}
        \widehat{B_{t_k}} = W_k D W^\dagger_k, \text{ where } D \coloneqq I\otimes B^\sharp - (B^\flat)^\top\otimes I. 
    \end{align}
    Note that the vectorized Hamiltonian can be written as 
    \begin{align}
        \widehat{H} &= \int^\infty_{-\infty} g(t) F(t)~\d t, \text{ where } F(t) \coloneqq (\widehat{B_{t}})^\dagger \widehat{B_{t}} = W_t D^\dagger D W_t^\dagger
    \end{align}
    We notice that the integrand $\|F(t)\| \le 4$ for any $-\infty \le t \le \infty$, and
    \begin{align}
        \BB^\dagger \BB = \sum^{N-1}_{k=0} h g(t_k) F(t_k),
    \end{align}
    Now, with the step size $h>0$ and the quadrature number $N$, we define 
    \begin{align}
        S_h \coloneqq h\sum_{k=-\infty}^{+\infty} g(kh) F(kh),\quad S^{[N]}_{h} = h\sum^{N-1}_{k=0} g(h(k-N/2)) F(h(k-N/2)) = \BB^\dagger \BB.
    \end{align}
    By the triangle inequality, we have
    \begin{align}
        \|S_h - S^{[N]}_h\| &\le h \left\|\sum_{|k| \ge N/2} g(kh) F(kh)\right\|\\
        &\le 8h \sum_{k\ge N/2} g(kh)\\
        &\le 8h \sum_{k\ge N/2} \frac{2}{\beta}\exp(-2\pi k h/\beta) \\
        &= \frac{16h}{\beta} \frac{e^{-\pi hN/\beta}}{1-e^{-2\pi h /\beta}}
    \end{align}
    Recall that $Nh = 2T$.
    Using $1-e^{-x}\ge x/2$ for $0\le x\le1$, we have 
    \begin{equation}\label{eqn:tail-estimate-1}
        \|S_h - S^{[N]}_h\| \le \frac{16 \exp(-2\pi T/\beta)}{\pi} \qquad \text{with} \quad 0 < h \le \frac{\beta}{2\pi}.
    \end{equation}

    It remains to bound the aliasing error $\|\widehat{H}-S_h\|$.
    Let 
    \begin{align}
        K \coloneqq -H^\top\otimes I + I\otimes H.
    \end{align}    
    Then $W_t = e^{itK}$, and 
    the orbit $t \mapsto F(t)$ has frequencies given by differences of eigenvalues of $K$. 
    Let $\lambda_u$ (or $\lambda'_v$) be the eigenvalues of $-H^\top$ (or $H$) with the corresponding eigen-projector $\Pi_u$ (or $\Pi'_v$). 
    The eigenvalues and eigenprojectors of $K$ are then $E_{u,v} = \lambda_u + \lambda'_v$ and $P_{u,v} = \Pi_u \otimes \Pi'_v$, and we can write $K = \sum_{u,v} E_{u,v}P_{u,v}$. Then, we have
    \begin{align}
        \widehat{H} &= \int^\infty_{-\infty} g(t) F(t)~\d t= \int^\infty_{-\infty} g(t) \sum_{\nu} e^{i t \nu} \left(\sum_{E_{u,v}-E_{u',v'}=\nu}P_{u,v} D^\dagger DP_{u',v'}\right)~\d t\\
        &= \sum_\nu \hat{g}(\nu) \left(\sum_{E_{u,v}-E_{u',v'}=\nu}P_{u,v} D^\dagger DP_{u',v'}\right).
    \end{align}
    By the Poisson summation formula, we have 
    \begin{equation}
        h\sum_k g(kh)e^{iykh} = \sum_n \hat{g}\left(y + \frac{2\pi n}{h}\right)
    \end{equation}
    for every frequency $y$. Since $\|D\| \le 2$, we have $\|D^\dagger D\| \le 4$. Moreover, all frequencies satisfy $|\nu| \le 4\|H\|$. Therefore, if
    \begin{equation}\label{eqn:aliasing-bound-h}
        h \le \frac{\pi}{4\|H\|},
    \end{equation}
    then $|\nu + 2\pi n/h| \ge \pi n/h$ for all $n\ge 1$. Using $\hat{g}(y) = \frac{1}{2\cosh(\beta y/4)} \le e^{-\beta|y|/4}$, we obtain
    \begin{align}
        \|\widehat{H} - S_h\| &\le 8 \sum^\infty_{n=1} \exp\left(-\frac{\beta \pi n}{4h}\right) = \frac{8}{\exp(\beta \pi/4h)-1}.
    \end{align}
    Combining~\cref{eqn:tail-estimate-1},~\cref{eqn:aliasing-bound-h}, and the preceding estimate, we can choose $T$, $h$ such that (assuming $0 < \epsilon < 1$)
    \begin{equation}
        T = \frac{\beta}{2\pi}\log\left(\frac{32}{\pi \epsilon}\right),\quad \frac{1}{h} = \max\left(\frac{4\log(17/\epsilon)}{\pi \beta},\frac{2\pi}{\beta},\frac{4\|H\|}{\pi}\right),
    \end{equation}
    and by the triangle inequality, we have $\|\widehat{H}-\BB^\dagger\BB\|\le \epsilon$. Since $N = 2T/h$, this implies 
    \begin{equation}
        N = \cO\left(\log^2\left(\frac{1}{\epsilon}\right) + \beta \|H\|\log\left(\frac{1}{\epsilon}\right) \right).
    \end{equation}
    For general $J$, write $\widehat{H} = \sum_{j=1}^J \widehat{H}^{(j)}$ and $\BB^\dagger\BB = \sum_{j=1}^J M^{(j)}$, where $\widehat{H}^{(j)}$ and $M^{(j)}$ are the contributions of the $j$-th jump operator. The same single-jump estimate applies to each pair $(\widehat{H}^{(j)},M^{(j)})$. Choosing the discretization so that $\|\widehat{H}^{(j)}-M^{(j)}\|\le \epsilon/J$ for every $j$, the triangle inequality yields $\|\widehat{H}-\BB^\dagger\BB\|\le \epsilon$. This concludes the proof.
\end{proof}

\begin{lem}\label{lem:singularperturbation}
    Suppose that $\epsilon \le \Delta/2$ and $\|\widehat{H} - \BB^\dagger \BB\| \le \epsilon$, and let $\ket{\psi_g}$ be a normalized right singular vector of $\BB$ corresponding to its smallest singular value. 
    Then, we have 
    \begin{align}
        1 - \left|\llangle \sr | \psi_g\rangle\right|^2 \le \left(\frac{2\epsilon}{\Delta}\right)^2,
    \end{align} 
    where $\Delta$ is the spectral gap of $\widehat{H}$.
\end{lem}
\begin{proof}
    Let $P_0 = |\sr\rrangle \llangle \sr |$ be the ground-state projector of $\widehat{H}$, and let $Q_0 = \ketbra{\psi_g}{\psi_g}$ be the ground-state projector of $\BB^\dagger\BB$.  By Weyl's inequality, if $\lambda_1(\widehat{H})=0$ and $\lambda_2(\widehat{H})\ge \Delta$, then the eigenvalues $\mu_1\le \mu_2\le \cdots$ of $\BB^\dagger\BB$ satisfy $0\le \mu_1\le \epsilon$ and $\mu_2\ge \Delta-\epsilon\ge \Delta/2$. Hence $\BB^\dagger\BB$ has a unique eigenvalue in $[0,\Delta/2)$, and $Q_0$ is the spectral projector onto the corresponding one-dimensional eigenspace.
    The spectra associated with $P_0$ and $I-Q_0$ are therefore contained in $\{0\}$ and $[\Delta/2,\infty)$, respectively, so their separation is at least $\Delta/2$. The Davis--Kahan theorem~\cite[Theorem VII.3.1]{Bhatia1997} then gives
    \begin{equation}
        1 - \left|\llangle \sr | \psi_g\rangle\right|^2 = \|P_0(I-Q_0)\|^2 \le \left(\frac{\|\widehat{H}-\BB^\dagger\BB\|}{\Delta/2}\right)^2 \le \left(\frac{2\epsilon}{\Delta}\right)^2.
    \end{equation}
\end{proof}

\subsection{Quantum circuit implementation}\label{append:circuit-implementation-cost}

\begin{defn}[Block-encoding of a rectangular matrix~\cite{GilyenSuLowEtAl2019}]\label{defn:be-general}
    Given a matrix $A\in \CC^{2^n\times 2^p}$ with $n\ge p$, if we can find $\alpha, \epsilon > 0$, and a unitary matrix $U_A\in \CC^{2^{n+m}\times 2^{n+m}}$ such that
    \begin{equation}
        \|A - \alpha \left(\bra{0^m}\otimes I_{2^n}\right) U_A \left(\ket{0^m}\otimes I_{2^p}\right)\| \le \epsilon,
    \end{equation}
    then $U_A$ is called an $(\alpha, m, \epsilon)$-block-encoding of $A$. 
\end{defn}

For efficiently implementable Lindbladians like CKG and DLL, both operators $B^\sharp$ and $B^\flat$ can be implemented using weighted operator Fourier transform, and the implementation cost is a similar; a detailed discussion can be found in the previous section.
For simplicity, we now assume access to block-encodings of these factor matrices.

Suppose that we have access to: 
\begin{enumerate}
    \item Block-encodings of the $\{(B^\flat_j)^\top\}$ and $\{B^\sharp_j\}$ operators, denoted as $U_{(B^\flat)^\top}$ and $U_{B^\sharp}$. 
    Suppose each jump operator acts on $n$ qubits, and there are $J = 2^r$ jump operators. The block-encodings use $p$ ancilla qubits such that 
        \begin{align}
            &\left(\bra{0^p}\otimes I_{2^{(r+n)}}\right) U_{(B^\flat)^\top} \left(\ket{0^p}\otimes I_{2^n}\right) = \frac{1}{C_\flat}\sum^{J-1}_{j=0} \ketbra{j}{0}\otimes (B^\flat_j)^\top,\\ 
            &\left(\bra{0^p}\otimes I_{2^{(r+n)}}\right) U_{B^\sharp}\left(\ket{0^p}\otimes I_{2^n}\right) = \frac{1}{C_\sharp}\sum^{J-1}_{j=0} \ketbra{j}{0}\otimes B^\sharp_j,
        \end{align}
    where $C_\sharp$ and $C_\flat$ are the sub-normalization factors of the two block-encodings. As discussed before, for efficiently implementable quantum Gibbs samplers like CKG or DLL, it is possible to have $\|B^\sharp_j\|, \|B^\flat_j\| = \cO(1)$ for all $j$. Therefore, the smallest sub-normalization factors consistent with unitarity scale as $\cO(\sqrt{J})$. Without loss of generality, we assume $C_\sharp = C_\flat = \sqrt{J}$.
    \item Select oracles for implementing controlled time-evolution of $H$ and $H^\top$, denoted as $U_{\rm sel}$ and $U^{\top}_{\rm sel}$. Suppose that the (time domain) discretization number is $N = 2^q$ and
    \begin{align}
        U_{\rm sel} = \sum^{N-1}_{j=0} \ketbra{j}{j}\otimes e^{it_j H},\quad U^{\top}_{\rm sel} = \sum^{N-1}_{j=0} \ketbra{j}{j}\otimes e^{it_j H^\top},\quad t_k = -T + \frac{2kT}{N}~\text{with}~0 \le k \le N-1. 
    \end{align} 
    \item A state preparation oracle $U_{\rm prep}$:
    \begin{align}
        U_{\rm prep}\ket{0^q} = \ket{g}, \quad \ket{g} := \frac{1}{\sqrt{C}} \sum^{N-1}_{k=0}\sqrt{g_k}\ket{k}, \label{eq:Uprep}
    \end{align}
    where $g_k := g(t_k)\Delta t$. 
    Note that, given our choice of $N$ and $T$ as in~\cref{lem:quadrature_error}, the normalization factor $C :=\sum^{N-1}_{k=0}g_k$ is a constant close to $\int_{-\infty}^{\infty} g(t)\d t = 1/2$.
\end{enumerate}

\begin{thm}
    We can prepare a $(2\sqrt{CJ}, p+1, 0)$-block-encoding of $\BB$ using one query to each of $U_{(B^\flat)^\top}$ and $U_{B^\sharp}$, two queries to each of $U_{\rm sel}$ and $U^{\top}_{\rm sel}$ (and their inverses), one query to $U_{\rm prep}$, and a constant number of elementary gates.
\end{thm}
\begin{proof}
    First, we block-encode $B^\sharp_{j,t}$ and $(B^\flat_{j,t})^\top$ separately. We define
    \begin{align}
        \tilde{U}_{B^\sharp} &:=  (I_{2^{(p+r)}} \otimes U_{\rm sel}) (I_{2^q}\otimes U_{B^\sharp})(I_{2^{(p+r)}} \otimes U^\dagger_{\rm sel}),\\
        \tilde{U}_{(B^\flat)^\top} &:=  (I_{2^{(p+r)}} \otimes (U^{\top}_{\rm sel})^\dagger) (I_{2^q}\otimes U_{(B^\flat)^\top})(I_{2^{(p+r)}} \otimes U^{\top}_{\rm sel}).
    \end{align}
    Direct calculation shows that 
    \begin{align}
        \left(\bra{0^p}\otimes I_{2^{(n+r+q)}}\right) \tilde{U}_{B^\sharp} \left(\ket{0^{p}}\otimes I_{2^{(n+q)}}\right) 
        = \frac{1}{\sqrt{J}}\sum^{2^r-1}_{j=0} \ketbra{j}{0}\otimes\left( \sum^{2^q-1}_{k=0} \ketbra{k}{k}\otimes e^{it_k H}B^\sharp_j e^{-it_k H}\right),\\
        \left(\bra{0^p}\otimes I_{2^{(n+r+q)}}\right) \tilde{U}_{(B^\flat)^\top} \left(\ket{0^{p}}\otimes I_{2^{(n+q)}}\right) 
        = \frac{1}{\sqrt{J}}\sum^{2^r-1}_{j=0} \ketbra{j}{0}\otimes\left( \sum^{2^q-1}_{k=0} \ketbra{k}{k}\otimes e^{-it_k H^\top}(B^\flat_j)^\top e^{it_k H^\top}\right).
    \end{align}
    
    Then, we add an extra ancilla qubit to $\tilde{U}_{B^\sharp}$ and $\tilde{U}_{(B^\flat)^\top}$, so their block-encoding parts are selected by $\ket{0}$ and $\ket{1}$, respectively. 
    Combining these two circuits and using the LCU lemma~\cite[Lemma~29]{GilyenSuLowEtAl2019},we obtain a $(2{\sqrt{J}},p+1,0)$-block-encoding of the following rectangular matrix with the vectorized jump operators:
    \begin{equation}
        \sum^{2^r-1}_{j=0} \ketbra{j}{0}\otimes \left[\sum^{2^q-1}_{k=0} \ketbra{k}{k} \otimes \left(I_{2^n}\otimes e^{it_k H}B^\sharp_j e^{-it_k H} - e^{-it_k H^\top}(B^\flat_j)^\top e^{it_k H^\top}\otimes I_{2^n} \right)\right],
    \end{equation}
    The state preparation oracle in the LCU can be built using $HX$ (which maps $\ket{0}$ to $(\ket{0}-\ket{1})/\sqrt{2}$).
    We denote this full circuit as $\tilde{U}_{B,t}$.

    Finally, we define 
    \begin{equation}
        U_{\BB} := \tilde{U}_{B,t} (I_{p+1+2n}\otimes U_{\rm prep}).
    \end{equation}
    We now check that $U_{\BB}$ is a $(2{\sqrt{JC}},p+1,0)$-block-encoding of $\BB$:
    \begin{align}
        &\left(\bra{0^{p+1}}\otimes I_{2^{(2n+r+q)}}\right) U_{\BB} \left(\ket{0^{p+1}}\otimes I_{2^{2n}}\right)\\ 
        &= \frac{1}{2\sqrt{J}}\sum^{J-1}_{j=0} \ketbra{j}{0^r} \otimes \left[\sum^{N-1}_{k=0} \ketbra{k}{k} \otimes \left(I_{2^n}\otimes e^{it_k H}B^\sharp_j e^{-it_k H} - e^{-it_k H^\top}(B^\flat_j)^\top e^{it_k H^\top} \otimes I_{2^n} \right)\right]\nonumber\\
        &\qquad\times\left(\frac{1}{\sqrt{C}}\sum^{N-1}_{k=0} \sqrt{g_k}\ketbra{k}{0^q}\right)\\
        &= \frac{1}{2\sqrt{JC}} \sum^{J-1}_{j=0} \ketbra{j}{0^r} \otimes \left[\sum^{2^q-1}_{k=0} \ketbra{k}{0^q} \otimes \sqrt{g_k}\left(I_{2^n}\otimes e^{it_k H}B^\sharp_j e^{-it_k H} - e^{-it_k H^\top}(B^\flat_j)^\top e^{it_k H^\top} \otimes I_{2^n} \right)\right]\\
        &= \frac{1}{2\sqrt{JC}}\sum^{J-1}_{j=0} \sum^{N-1}_{k=0} \ketbra{j,k}{0^{r+q}} \otimes \left[\sqrt{g_k}\left(I_{2^n}\otimes e^{it_k H}B^\sharp_j e^{-it_k H} - e^{-it_k H^\top}(B^\flat_j)^\top e^{it_k H^\top} \otimes I_{2^n} \right)\right].
    \end{align}
    Therefore, $U_{\BB}$ is a block-encoding of $\BB$ with a sub-normalization factor $\frac{1}{\sqrt{JC}}$, equivalently a $(2\sqrt{JC},p+1,0)$-block-encoding.
    In this construction, we use one query to each of $U_{(B^\flat)^\top}$ and $U_{B^\sharp}$; two queries to each of $U_{\rm sel}$ and $U^{\top}_{\rm sel}$ (and their inverses), and one query to $U_{\rm prep}$.
    The number of additional elementary gates is a constant. 
\end{proof}

\subsection{Complexity analysis}\label{sec:complexity}
Now, we are ready to state the main result.

\begin{thm}
    Suppose that we have access to:
    \begin{enumerate}
      \item a warm-start state $\ket{\phi}$ such that $|\langle \phi |\sr\rrangle |^2=\Omega(1)$, 
      \item $(\sqrt{J},p)$-block-encodings of $\{(B^\flat_j)^\top\}^{J-1}_{j=0}$ and $\{B^\sharp_j\}^{J-1}_{j=0}$, denoted by $U_{(B^\flat)^\top}$ and $U_{B^\sharp}$, 
      \item select oracles $U_{\rm sel}$ and $U^{\top}_{\rm sel}$  for controlled time-evolution under $H$ and $H^\top$, respectively, 
      \item a state preparation oracle $U_{\rm prep}$ as in Eq.~(\ref{eq:Uprep}).
    \end{enumerate}
    Let $\Delta = \textrm{Gap}(\cL^\dagger)$ denote the Lindbladian gap, and $J$ is the total number of jump operators. Then, there exists a quantum algorithm that outputs an $\epsilon$-approximate of $|\sr\rrangle$ using
    \begin{itemize}
        \item $\cO\left(\sqrt{\frac{J}{\Delta}}\cdot \log(\frac{1}{\epsilon})\right)$ quantum queries to $U_{(B^\flat)^\top}$, $U_{B^\sharp}$, $U_{\rm sel}$, $U^{\top}_{\rm sel}$, 
        $U_{\rm prep}$ (and their inverse or controlled versions),
        \item $\cO(1)$ uses of the warm-start state $\ket{\phi}$,
        \item The total number of qubits is $\cO\left(n + \log(\beta)+\log(J)+\log(\|H\|)+\log\log(J/\Delta \epsilon)\right)$,
        \item The maximal Hamiltonian evolution time is $T_{\max} = \cO\left(\beta \log(J/\Delta \epsilon)\right)$.
    \end{itemize}
\end{thm}

\begin{proof}
    Let $\eta := \epsilon\Delta/2$ with $\epsilon<1$. By~\cref{lem:quadrature_error}, we can choose $T$ and $N$ so that
    \begin{equation}
        \|\widehat{H}-\BB^\dagger\BB\| \le \eta,
    \end{equation}
    with
    \begin{equation}
        T = \cO\left(\beta\log\frac{J}{\Delta\epsilon}\right),\qquad
        N = \cO\left(\log^2\frac{J}{\Delta\epsilon} + \beta\|H\|\log\frac{J}{\Delta\epsilon}\right).
    \end{equation}
    By the previous theorem, $\BB$ admits a block-encoding with a constant subnormalization factor, and so does $\BB^\dagger\BB$.

    Let $\ket{\psi_g}$ be the right singular vector of $\BB$ corresponding to its smallest singular value. By \cref{lem:singularperturbation},
    \begin{equation}
        1-|\llangle \sr|\psi_g\rangle|^2 \le \left(\frac{2\eta}{\Delta}\right)^2 \le \epsilon^2.
    \end{equation}
    Moreover, the second eigenvalue of $\BB^\dagger\BB$ is at least $\Delta-\eta \ge \Delta/2$, so the nonzero singular values of $\BB$ are separated from $0$ by at least $\sqrt{\Delta/2}$.

    We may therefore apply singular value thresholding to the block-encoding of $\BB$. Note that the block-encoding has a sub-normalization factor $2\sqrt{JC}=\cO(\sqrt{J})$.
    Starting from the warm-start state $\ket{\phi}$, the analysis of~\cite[Proposition 5, SI]{leng2026operator} yields an $\epsilon$-approximation of $|\sr\rrangle$ using
    \begin{equation}
        \cO\left(\sqrt{\frac{J}{\Delta}}\log\frac{1}{\epsilon}\right)
    \end{equation}
    queries to the block-encoding of $\BB$, and hence the same number of queries to $U_{(B^\flat_j)^\top}$, $U_{B^\sharp_j}$, $U_{\rm sel}$, $U_{\rm sel}^{\top}$, and $U_{\rm prep}$, up to constant factors. The qubit count follows from $r=\log J$ and $q=\log N$, and the maximal evolution time is $T_{\max}=T$. 
\end{proof}

\begin{rem}
The warm-start assumption can be relaxed to $|\langle \phi |\sr\rrangle |^2=\Omega(1/\poly(n))$, which would still yield a polynomial-time quantum algorithm.
\end{rem}

\section{Slow Gibbs samplers with efficiently preparable warm starts}
\label{append:easy-warm-starts}

We construct a family of inhomogeneous Curie--Weiss models whose thermofield double states admit efficiently preparable approximations at fixed trace-distance accuracy.
Nevertheless, a specified Gaussian Gibbs sampler built from single-site Pauli coupling operators requires time exponential in the system size to refine the corresponding warm-start states to a smaller error.
These examples thus provide concrete warm starts for our quantum algorithms, which offer a quadratic improvement in spectral-gap dependence.

Throughout this section, $d(\rho,\eta) = \frac{1}{2}\|\rho-\eta\|_1$ denotes the trace distance, and $D(\rho\|\eta) = \Tr(\rho\log \rho - \rho \log\eta)$ denotes the relative entropy, with natural logarithms. For pure states, $d$ refers to the corresponding density matrices.

\subsection{Hamiltonian models and easy warm starts}
Let $n$ be odd and let $X_i$ and $Z_i$ be the Pauli X and Z operators acting on the $i$-th qubit.
We consider inhomogeneous Curie--Weiss Hamiltonians of the following form,
\begin{align}\label{eq:warm-start-hamiltonian}
  H = H_{\rm cl} + V, \qquad \text{where}\quad H_{\rm cl} = - \frac{1}{2n}\sum^n_{i,j=1}Z_i Z_j - \frac{1}{n}\sum^n_{i=1}Z_i,\quad V = - \sum_{i< j}J_{i,j} X_i X_j.
\end{align}
Here, the coefficients $J_{i,j} = J_{j,i}$ and $0 < J_{i,j} \le \kappa /n$ for some $\kappa$ that will be specified later.
The $1/n$ normalization keeps the dense two-body interactions extensive, and for sufficiently small $\kappa$, the quantum model retains the classical critical inverse temperature $\beta_c=1$.
The off-diagonal entries in each column of $J$ are taken be pairwise distinct.\footnote{This additional assumption is used only to rule out bounded-depth exact diagonalization; it is not needed for the warm-start or slow-refinement bounds.}
Given an inverse temperature $\beta > 0$, we define the Gibbs states associated with $H_{\rm cl}$ and $H$, respectively,
\begin{align}
  \rhocl := \frac{e^{-\beta \Hcl}}{Z_{\rm cl}}, \qquad \sigma := \frac{e^{-\beta H}}{Z_H},
\end{align}
where $Z_{\rm cl} := \Tr(e^{-\beta \Hcl})$ and $Z_H := \Tr(e^{-\beta H})$ are the partition functions.

The classical Hamiltonian $\Hcl$ is the ferromagnetic Curie--Weiss model with a weak longitudinal field.
Writing $M=\sum_{i=1}^n Z_i$, the interaction term $-M^2/(2n)$ has two fully polarized ground states, $\ket{0^{\otimes n}}$ and $\ket{1^{\otimes n}}$, while the field term $-M/n$ breaks spin-flip symmetry and favors positive magnetization.
For fixed $\beta>1$, the distribution of the magnetization density $M/n$ under $\rhocl$ concentrates, as $n\to\infty$, near $\pm m^*$, where $m^*\in(0,1)$ is the unique positive solution of the Curie--Weiss self-consistency equation
\begin{align}
    m^* = \tanh(\beta m^*).
\end{align}
The field vanishes per spin in this limit, so it does not change the leading-order peak locations.
Its order-one contribution to the total energy nevertheless biases the phase weights: their positive-to-negative ratio tends to $e^{2\beta m^*}$.
Thus, at sufficiently low temperature, the positive phase carries most of the Gibbs weight.

To approximate this dominant phase, we choose $a:=\beta m^*$ and define
\begin{align}
  \tau := \frac{e^{aM}}{Z_\tau}, \qquad Z_\tau := (2\cosh a)^n.
\end{align}
This product state replaces the correlations within the positive classical phase by independent spins with the same leading-order magnetization density,
\begin{align}
  \Tr\left(\frac{M}{n}\tau\right)=\tanh(a)=m^*.
\end{align}
Its canonical purification is
\begin{align}
  \ket{\Phi}:=|\sqrt{\tau}\rrangle
  =\bigotimes_{i=1}^n\left(
  \sqrt{\frac{1+m^*}{2}}\ket{00}_{L_iR_i}
  +\sqrt{\frac{1-m^*}{2}}\ket{11}_{L_iR_i}\right).
\end{align}
Each factor involves only one pair of qubits, making this purification easy to prepare.

We quantify this approximation using the relative entropy.
Since $V$ has zero diagonal in the computational basis,
\begin{align}\label{eq:warm-start-entropy-decomposition}
  D(\tau\|\sigma)=D(\tau\|\rhocl)+\log\frac{Z_H}{Z_{\rm cl}}.
\end{align}
The first term becomes small at sufficiently low temperature, while the second is bounded by $\beta^2\kappa^2$ whenever $\beta\kappa\leq1/2$, uniformly in $n$.
Thus, for a sufficiently large fixed $\beta$, choosing $\kappa$ to be a small constant multiple of $1/\beta$ controls both contributions and yields the following warm-start guarantee.

\begin{lem}[Efficient thermofield-double warm start]\label{lem:easy-warm-start}
  There exist constants $\beta_0>1$ and $c>0$ such that the following holds for every fixed $\beta\geq\beta_0$.
  Let $m^*\in(0,1)$ be the unique positive solution of $m^*=\tanh(\beta m^*)$, and set $a=\beta m^*$.
  For all sufficiently large odd $n$ and all coupling matrices whose off-diagonal entries satisfy
  \begin{align}
    0<J_{i,j}=J_{j,i}\leq\frac{\kappa}{n},
    \qquad 0<\kappa\leq\frac{c}{\beta},
  \end{align}
  the canonical purifications obey
  \begin{align}
    |\langle\Phi|\sqrt{\sigma}\rrangle|^2>0.99.
  \end{align}
  Moreover, $\ket{\Phi}$ can be prepared from $\ket{0}^{\otimes 2n}$ using $\cO(n)$ gates and $\cO(1)$ depth, with arbitrary single-qubit rotations and direct two-qubit gates within each pair $(L_i,R_i)$.
\end{lem}

\begin{proof}
  Write $m=m^*$, $v=1-m^2$, and
  \begin{align}
    g(y)=\log(2\cosh y)-\frac{y^2}{2\beta}.
  \end{align}
  We first estimate the classical contribution in~\eqref{eq:warm-start-entropy-decomposition}.
  Gaussian integration gives the exact representation
  \begin{align}\label{eq:warm-start-classical-partition}
    Z_{\rm cl}
    =e^{-\beta/(2n)}\sqrt{\frac{n}{2\pi\beta}}
    \int_{\mathbb R}e^{ng(y)+y}\,\d y.
  \end{align}
  For fixed $\beta>1$, the function $g$ has exactly two global maxima, at $\pm a$.
  Indeed, $g'(y)=\tanh y-y/\beta$, and $\tanh y/y$ decreases strictly on $(0,\infty)$.
  The maxima are nondegenerate, with $g''(\pm a)=v-1/\beta<0$.
  Since $g$ also decays quadratically as $|y|\to\infty$, Laplace's method applied at these two maxima yields
  \begin{align}
    Z_{\rm cl}
    =e^{ng(a)}\frac{e^a+e^{-a}}{\sqrt{1-\beta v}}\bigl(1+o(1)\bigr),
    \qquad n\to\infty.
  \end{align}
  The product-state moments $\Tr(\tau M)=nm$ and $\Tr(\tau M^2)=n^2m^2+nv$ imply
  \begin{align}
    D(\tau\|\rhocl) &=\log Z_{\rm cl}-ng(a)-a-\frac{\beta v}{2} =r_\beta+o(1),\label{eq:warm-start-classical-entropy}\\
    r_\beta&:=\log(1+e^{-2a})-\frac12\log(1-\beta v)-\frac{\beta v}{2}.
  \end{align}
  As $\beta\to\infty$, we have $a\to\infty$ and
  $\beta v=(a/\tanh a)\operatorname{sech}^2a\to0$; hence $r_\beta\to0$.
  The remainder in~\eqref{eq:warm-start-classical-entropy} is taken at fixed $\beta$ and is independent of the couplings.

  Next, suppose that $\beta\kappa\leq1/2$.
  The Golden--Thompson inequality implies
  \begin{align}
    \frac{Z_H}{Z_{\rm cl}} \leq \frac{1}{Z_{\rm cl}} \Tr\left(e^{-\beta H_{\rm cl}}e^{-\beta V}\right) = \Tr(\rhocl e^{-\beta V}).
  \end{align}
  Note that $e^{-\beta V}$ has a constant diagonal in the computational basis,
  \[
    e^{-\beta V} = \sum_{\xi \in \{\pm1\}^n} \exp(\beta \sum_{i<j}J_{i,j}\xi_i\xi_j) \ket{\xi}\bra{\xi},
  \]
  therefore
  \begin{align}
    \Tr(\rhocl e^{-\beta V}) = \mathbb E_\xi\exp\!\left(\beta\sum_{i<j}J_{i,j}\xi_i\xi_j\right),
  \end{align}
  where the $\xi_i$ are independent uniform signs.
  Choose a random partition $S\cup T=\{1,\ldots,n\}$ by assigning each site independently to either set with probability $1/2$.
  Every edge crosses the partition with probability $1/2$, so Jensen's inequality bounds the last expectation by
  \begin{align}
    \mathbb E_{S,T}\mathbb E_\xi
    \exp\!\left(2\beta\xi_S^TB\xi_T\right),
    \qquad B=J_{S,T}.
  \end{align}
  For a fixed partition, let $G$ be a standard Gaussian vector in $\mathbb R^{|S|}$.
  Applying $\cosh u\leq e^{u^2/2}$ to the two sign averages gives
  \begin{align}
    \mathbb E_\xi e^{2\beta\xi_S^TB\xi_T} \leq\mathbb E_{\xi_T}e^{2\beta^2\norm{B\xi_T}^2}
    =\mathbb E_{G,\xi_T}e^{2\beta G^TB\xi_T} \leq\mathbb E_Ge^{2\beta^2\norm{B^TG}^2}
    =\det(I-4\beta^2BB^T)^{-1/2}.
  \end{align}
  Here $\norm{B}_F^2\leq\kappa^2|S||T|/n^2\leq\kappa^2/4$, where $\norm{\cdot}_F$ is the Frobenius norm.
  Thus every eigenvalue of $4\beta^2BB^T$ lies in $[0,1/4]$.
  Using $-\log(1-x)\leq2x$ on this interval, we obtain
  \begin{align}
    \log\det(I-4\beta^2BB^T)^{-1/2}
    \leq4\beta^2\norm{B}_F^2\leq\beta^2\kappa^2.
  \end{align}
  This estimate is uniform over the partitions and the allowed couplings.
  Also, $\Tr(\rhocl V)=0$ implies
  $\log(Z_H/Z_{\rm cl})=D(\rhocl\|\sigma)\geq0$.
  Combining these bounds with~\eqref{eq:warm-start-entropy-decomposition} gives
  \begin{align}\label{eq:warm-start-total-entropy}
    D(\tau\|\sigma)\leq r_\beta+\beta^2\kappa^2+o(1).
  \end{align}

  Set $\varepsilon=-\log(0.99)>0$.
  Choose $\beta_0>1$ such that $r_\beta<\varepsilon/3$ for all $\beta\geq\beta_0$, and choose a constant $c>0$ with $c\leq1/2$ and $c^2\leq\varepsilon/3$.
  For any fixed $\beta\geq\beta_0$, take $n$ sufficiently large that the absolute value of the remainder in~\eqref{eq:warm-start-total-entropy} is less than $\varepsilon/3$.
  Then $0<\kappa\leq c/\beta$ ensures $D(\tau\|\sigma)<\varepsilon$, uniformly over the allowed couplings.

  To convert this entropy bound into a canonical-purification overlap, let
  $\tau=\sum_i p_i\ket{u_i}\bra{u_i}$ and
  $\sigma=\sum_j q_j\ket{v_j}\bra{v_j}$ be spectral decompositions.
  Both states are faithful, and the weights $p_i|\langle u_i|v_j\rangle|^2$ sum to one.
  Jensen's inequality therefore gives
  \begin{align}
    \Tr(\sqrt{\tau}\sqrt{\sigma}) =\sum_{i,j}p_i|\langle u_i|v_j\rangle|^2 \exp\!\left(-\frac12\log\frac{p_i}{q_j}\right) \geq\exp\!\left(-\frac12D(\tau\|\sigma)\right).
  \end{align}
  Since $\langle\Phi|\sqrt{\sigma}\rrangle=\Tr(\sqrt{\tau}\sqrt{\sigma})$, it follows that
  \begin{align}
    |\langle\Phi|\sqrt{\sigma}\rrangle|^2
    \geq e^{-D(\tau\|\sigma)}>e^{-\varepsilon}=0.99.
  \end{align}

  Finally, prepare each pair factor of $\ket{\Phi}$ by applying a single-qubit rotation to $L_i$, followed by a CNOT from $L_i$ to $R_i$.
  All rotations can be performed in parallel, as can all CNOTs, so the circuit uses $2n$ gates in two layers in the stated gate model.
\end{proof}

\subsection{Gibbs samplers and slow mixing}\label{sec:slow-gibbs-mixing}

We consider the Gaussian CKG Gibbs sampler~\cite{chen2023efficient} with coupling operators
\begin{align}\label{eq:slow-ckg-couplings}
  A_j=\frac{S_j}{\sqrt{3n}},
  \qquad \{S_j\}_{j=1}^{3n}=\{X_i,Y_i,Z_i:1\leq i\leq n\}.
\end{align}
We apply the $1/\sqrt{3n}$ normalization so $\sum_j A_j^\dagger A_j=I$.
In the CKG sampler, the jump operators are given by
\begin{align}\label{eq:slow-ckg-filters}
  \widehat A_j(\omega)
  &=\frac{1}{\sqrt{2\pi}}\int_{\mathbb R}
    f_\beta(t)e^{-i\omega t}e^{itH}A_je^{-itH}\,\d t,\\
  f_\beta(t)&=\left(\frac{2}{\pi\beta^2}\right)^{1/4}e^{-t^2/\beta^2},
  \qquad \gamma_\beta(\omega)=e^{-(\beta\omega+1)^2/2}.\nonumber
\end{align}
Writing $\cD_L(\rho)=L\rho L^\dagger-\frac12\{L^\dagger L,\rho\}$, the CKG Lindbladian is
\begin{align}\label{eq:slow-ckg-generator}
  \cL_{\rm CKG}(\rho)
  =-i[G,\rho]+\sum_{j=1}^{3n}\int_{\mathbb R}
    \gamma_\beta(\omega)\cD_{\widehat A_j(\omega)}(\rho)\,\d\omega.
\end{align}
The coherent term is specified in an energy eigenbasis $H\ket r=E_r\ket r$ by
\begin{align}
  G_{rs}=\frac{i}{2}\tanh\!\left(\frac{\beta(E_r-E_s)}4\right)R_{rs},\quad R=\sum_{j=1}^{3n}\int_{\mathbb R}\gamma_\beta(\omega)
    \widehat A_j(\omega)^\dagger\widehat A_j(\omega)\,\d\omega.\nonumber
\end{align}

By~\cite[Theorem~I.1]{chen2023efficient}, this generator satisfies exact KMS detailed balance with faithful stationary state $\sigma$.
The coherent-term estimate in~\cite[Corollary~III.1]{chen2023efficient}, together with Parseval's identity, $\gamma_\beta\leq1$, and our coupling normalization, gives
\begin{align}\label{eq:slow-ckg-normalization}
  \norm{G}\leq1,
  \qquad\norm{\cL_{\rm CKG}^\dagger}_{\infty\to\infty}\leq4,
\end{align}
where $\norm{\mathcal A}_{\infty\to\infty}:=\sup_{\norm{X}\leq1}\norm{\mathcal A(X)}$.
Primitivity follows from the commutant criterion in~\cite[Lemmas~2--3]{ding2025efficient}: since $\gamma_\beta(\omega)>0$ and $f_\beta(0)\ne0$, Fourier inversion at $t=0$ shows that any operator commuting with all filtered jumps also commutes with every $A_j$.
The single-site Pauli couplings have only scalar operators in their common commutant.

The slow relaxation is associated with the phase-sign observable $F=\operatorname{sgn}(M)$, which distinguishes positive and negative magnetization.
For odd $n$, $F^2=I$, and the two phases are separated by the boundary sector $P_{\rm b}=\mathbf{1}_{\{|M|\leq4\}}$.
At low temperature, this sector has exponentially small Gibbs weight, while both phases retain weights bounded below independently of $n$ at fixed $\beta$.
The bare single-site couplings change $M$ by at most two, and the pair-flip terms in $V$ change it by at most four.
Controlling the additional spreading caused by frequency filtering shows that transfer between the phases remains exponentially suppressed.
Thus $F$ is an approximate slow mode associated with this magnetization bottleneck.

We write $\cE_\sigma(X):=\cE(X,X)=-\langle X,\cL_{\rm CKG}^\dagger(X)\rangle_\sigma$ for this sampler.
Under the parameter assumptions below, the sum-of-squares representation and the bottleneck estimates in~\cref{lem:slow-filtered-phase-flux} give a constant $\eta_\beta>0$ such that, for all sufficiently large odd $n$,
\begin{align}\label{eq:slow-ckg-dirichlet-input}
  \cE_\sigma(F)\leq e^{-\eta_\beta n},
\end{align}
uniformly over the allowed couplings.
For the centered observable $F_0=F-\Tr(\sigma F)I$, \cref{lem:slow-kms-variance} bounds the KMS variance $\langle F_0,F_0\rangle_\sigma$ below by a positive constant depending only on the fixed temperature.
The Rayleigh quotient
\begin{align}\label{eq:slow-ckg-rayleigh}
  \Delta(\cL_{\rm CKG})
  \leq\frac{\cE_\sigma(F)}{\langle F_0,F_0\rangle_\sigma}
\end{align}
therefore gives an exponentially small spectral gap, as summarized in the following theorem.

\begin{thm}[Exponentially slow mixing]\label{thm:slow-ckg-mixing}
  Choose the constants in~\cref{lem:easy-warm-start} with $\beta_0\geq3$ and $c\leq1/20$.
  For every fixed $\beta\geq\beta_0$, there exist constants $\alpha_\beta>0$ and $n_0(\beta)$ such that the following holds for every odd $n\geq n_0(\beta)$ and every coupling matrix satisfying
  \begin{align}
    0<J_{i,j}=J_{j,i}\leq\frac{\kappa}{n}\quad(i\ne j),
    \qquad 0<\kappa\leq\frac{c}{\beta}.
  \end{align}
  The normalized generator~\eqref{eq:slow-ckg-generator} is primitive with stationary state $\sigma$, and its spectral gap satisfies
  \begin{align}\label{eq:slow-ckg-gap}
    0<\Delta(\cL_{\rm CKG})\leq e^{-\alpha_\beta n}.
  \end{align}
\end{thm}

\begin{proof}
  As established above, the normalized sampler is primitive with stationary state $\sigma$ and satisfies exact KMS detailed balance.
  The direct sum-of-squares estimate in~\cref{lem:slow-filtered-phase-flux} gives
  \begin{align}
    \cE_\sigma(F)\leq e^{-\eta_\beta n}
  \end{align}
  for some $\eta_\beta>0$ and all sufficiently large odd $n$.
  For $F_0=F-\Tr(\sigma F)I$, \cref{lem:slow-kms-variance} gives
  $\langle F_0,F_0\rangle_\sigma\geq2b_\beta(1-b_\beta)>0$, where $b_\beta=\frac12e^{-2\beta}$.
  Since $F_0$ is KMS-orthogonal to $I$ and $\cE_\sigma(F_0)=\cE_\sigma(F)$, the Rayleigh quotient yields
  \begin{align}
    0<\Delta(\cL_{\rm CKG})
    \leq\frac{\cE_\sigma(F)}{\langle F_0,F_0\rangle_\sigma}
    \leq\frac{e^{-\eta_\beta n}}{2b_\beta(1-b_\beta)}
    \leq e^{-\alpha_\beta n}.
  \end{align}
  One may take $\alpha_\beta=\eta_\beta/2$ after increasing the size threshold.
  All constants and size thresholds are uniform over the allowed couplings.
\end{proof}

\begin{rem}
  In particular, for every fixed $0<\epsilon<1/2$, its worst-case trace-distance mixing time
  obeys
  \begin{align}
    t_{\rm mix}(\epsilon) \geq e^{\alpha_\beta n}\log\frac{1}{2\epsilon}.
  \end{align}
  Write $\Delta=\Delta(\cL_{\rm CKG})$ and choose a Hermitian gap eigenobservable $O$ with $\cL_{\rm CKG}^\dagger(O)=-\Delta O$ and $\norm{O}=1$.
  Stationarity implies $\Tr(\sigma O)=0$.
  An eigenstate $\rho_0$ of $O$ with eigenvalue of absolute value one then satisfies
  \begin{align}
    d\bigl(e^{t\cL_{\rm CKG}}(\rho_0),\sigma\bigr)
    \geq\frac12\left|\Tr\left[O\bigl(e^{t\cL_{\rm CKG}}(\rho_0)-\sigma\bigr)\right]\right|
    =\frac12e^{-\Delta t}.
  \end{align}
  Thus $t_{\rm mix}(\epsilon)\geq\Delta^{-1}\log(1/(2\epsilon))$, which gives the stated mixing-time lower bound.
  All size thresholds may depend on the fixed $\beta$, but are uniform over the allowed couplings.
\end{rem}

A small spectral gap alone gives a worst-case mixing time, but it does not rule out a fast thermal relaxation from specific initial states, such as the warm start $\ket{\Phi}$ from~\cref{lem:easy-warm-start}.
In the following, we prove that the warm start $\ket{\Phi}$ does not eliminate this obstruction to thermal refinement.
Its physical marginal $\tau=\Tr_R\ket{\Phi}\bra{\Phi}$ is concentrated in the positive-magnetization phase, with exponentially small weight in the negative phase.
The Gibbs state $\sigma$, however, assigns the negative phase a positive weight depending on $\beta$ but not vanishing with $n$.
Improving the warm marginal beyond a sufficiently small fixed error therefore requires restoring this missing population through the same bottleneck.
In~\Cref{thm:slow-ckg-warm-refinement}, we formally establish slow refinement specifically from $\tau$ under the CKG dynamics.

\begin{thm}[Exponentially slow thermal refinement of the warm start]\label{thm:slow-ckg-warm-refinement}
  Under the assumptions of~\cref{thm:slow-ckg-mixing}, let $\tau$ and $\ket{\Phi}=|\sqrt\tau\rrangle$ be the warm states defined above.
  For every fixed $\beta\geq\beta_0$, there exist constants $\epsilon_\beta\in(0,0.1)$ and $\nu_\beta>0$, independent of $n$ and of the individual allowed couplings, such that, for all sufficiently large odd $n$,
  \begin{align}\label{eq:slow-ckg-warm-refinement}
    d\bigl(e^{t\cL_{\rm CKG}}(\tau),\sigma\bigr)
    \geq\epsilon_\beta
    \qquad\text{for all }0\leq t\leq e^{\nu_\beta n}.
  \end{align}
  Consequently, thermal refinement of this marginal to any error smaller than $\epsilon_\beta$ requires time exponential in $n$.
\end{thm}

\begin{proof}
  Recall the two-phase observable $F=\operatorname{sgn}(M)$ and its spectral projectors $P_\pm=(I\pm F)/2$, which distinguish positive and negative magnetization.
  Quantum trace distance is at least the difference between the probabilities that two states assign to any chosen measurement outcome.
  We choose the negative-phase outcome $P_-$: the warm state $\tau$ is concentrated in the positive phase, whereas $\sigma$ has a nonvanishing negative-phase weight, and transitions between the phases are precisely the dynamical bottleneck.
  Writing $\mathcal T_t=e^{t\cL_{\rm CKG}^\dagger}$, measuring using $\{P_-,P_+\}$ gives the lower bound on the trace distance,
  \begin{align}
    d\bigl(e^{t\cL_{\rm CKG}}(\tau),\sigma\bigr) \geq\Tr(\sigma P_-)-\Tr\bigl(e^{t\cL_{\rm CKG}}(\tau)P_-\bigr) =\Tr(\sigma P_-)-\Tr\bigl(\tau\mathcal T_t(P_-)\bigr).
  \end{align}
  We will prove two bounds on the terms on the right-hand side.
  First, the stationary negative-phase weight satisfies
  \begin{align}\label{eq:warm-refinement-phase-weight}
    \Tr(\sigma P_-)\geq b_\beta:=\frac12e^{-2\beta}.
  \end{align}
  Second, with $D_0:=-\log(0.99)$, there is a constant $\nu_\beta>0$ such that, for all sufficiently large odd $n$,
  \begin{align}\label{eqn:negative-phase-population}
    \Tr\bigl(\tau\mathcal T_t(P_-)\bigr)
    \leq\frac{D_0+\log2}{\nu_\beta n}+(\cosh a)^{-n}
    \qquad\text{for }0\leq t\leq e^{\nu_\beta n}.
  \end{align}
  These bounds imply, throughout the same time interval,
  \begin{align}\label{eq:trace-distance-lower-bound}
    d\bigl(e^{t\cL_{\rm CKG}}(\tau),\sigma\bigr)
    \geq b_\beta-\frac{D_0+\log2}{\nu_\beta n}-(\cosh a)^{-n}
    \geq\frac{b_\beta}{2}
  \end{align}
  for sufficiently large $n$, since $a>0$ is fixed.
  Thus $\epsilon_\beta=b_\beta/2=e^{-2\beta}/4\in(0,0.1)$ gives the assertion.
  It remains to prove the two bounds, with all constants and size thresholds uniform over the allowed couplings.

  To prove~\eqref{eq:warm-refinement-phase-weight}, remove the tilt $-M/n$ from $H$.
  This restores global spin-flip symmetry, so each phase has Gibbs weight $1/2$ for odd $n$.
  Since $J_{i,j}>0$, all computational-basis Trotter path weights are nonnegative.
  Restoring the tilt multiplies each complete path weight by a factor in $[e^{-\beta},e^\beta]$, because $-I\leq M/n\leq I$.
  Hence the unnormalized negative-phase weight is at least $e^{-\beta}$ times its untilted value, while the partition function is at most $e^\beta$ times its untilted value.
  Dividing gives $\Tr(\sigma P_-)\geq e^{-2\beta}/2$, as claimed.

  To prove~\eqref{eqn:negative-phase-population}, consider the positive contraction
  \begin{align}
    E_t=P_+\mathcal T_t(P_-)P_+,
    \qquad p_t=\Tr(\sigma E_t),\quad q_t=\Tr(\tau E_t).
  \end{align}
  This measurement effect detects a transition from the positive phase to the negative phase over time $t$.
  Since $[\tau,F]=0$ and $0\leq\mathcal T_t(P_-)\leq I$, we have
  \begin{align}
    \Tr\bigl(\tau\mathcal T_t(P_-)\bigr)=q_t+\Tr\bigl(P_-\tau P_-\mathcal T_t(P_-)\bigr) \leq q_t+\Tr(\tau P_-)
    \leq q_t+(\cosh a)^{-n},
  \end{align}
  where the last step uses $P_-\leq e^{-aM}$ and $\Tr(\tau e^{-aM})=(\cosh a)^{-n}$.
  It therefore suffices to bound $q_t$.
  We first show that $p_t$ is exponentially small and then use the entropy bound $D(\tau\|\sigma)\leq D_0$ from the proof of~\cref{lem:easy-warm-start} to control $q_t$.

  In addition to the Dirichlet-form bound in~\cref{lem:slow-filtered-phase-flux}, we need to control $w_n:=\norm{[\sqrt\sigma,F]}_2$, where $\norm{\cdot}_2$ is the Hilbert--Schmidt norm.
  The boundary estimate~\eqref{eq:slow-sos-boundary} from the proof of that lemma, together with
  $[H,F]=P_{\rm b}[H,F]P_{\rm b}$ and $\norm{[H,F]}\leq\kappa n$, gives
  \begin{align}
    \norm{[H,F]\sqrt\sigma}_2\leq\kappa n\sqrt{z_n^2+r_n}.
  \end{align}
  In an energy eigenbasis, the Gibbs eigenvalues $p_r=e^{-\beta E_r}/Z_H$ satisfy
  $(\sqrt{p_r}-\sqrt{p_s})^2\leq\frac{\beta^2}{8}(E_r-E_s)^2(p_r+p_s)$.
  Summing this inequality against $\abs{F_{rs}}^2$ yields
  \begin{align}\label{eq:warm-refinement-coherence}
    w_n\leq\frac{\beta}{2}\norm{[H,F]\sqrt\sigma}_2
    \leq\frac{\beta\kappa n}{2}\sqrt{z_n^2+r_n}.
  \end{align}
  Since $z_n$ and $r_n$ decay exponentially, combining this with~\cref{lem:slow-filtered-phase-flux} gives a common constant $\zeta_\beta>0$ such that
  \begin{align}\label{eq:slow-ckg-bottleneck-inputs}
    \cE_\sigma(F)\leq e^{-\zeta_\beta n},
    \qquad w_n\leq e^{-\zeta_\beta n}
  \end{align}
  for all sufficiently large odd $n$, uniformly over the allowed couplings.
  To bound $p_t$, put $A=P_+\sqrt\sigma$ and observe that
  \begin{align}
    \norm{P_+\sigma P_+-\sqrt\sigma P_+\sqrt\sigma}_1 =\norm{AA^\dagger-A^\dagger A}_1 \leq2\norm{A-A^\dagger}_2\norm{A}_2\leq w_n.
  \end{align}
  Consequently,
  \begin{align}\label{eq:warm-refinement-crossing}
    p_t \leq w_n+\langle P_+,\mathcal T_t(P_-)\rangle_\sigma \leq w_n+\frac{w_n^2}{8}+\frac{t}{4}\cE_\sigma(F).
  \end{align}
  For the second inequality, use
  $\Tr(\sigma P_+)-\langle P_+,P_+\rangle_\sigma=w_n^2/8$
  and apply $1-e^{-t\lambda}\leq t\lambda$ in the KMS spectral decomposition of $-\cL_{\rm CKG}^\dagger$; also, $\cE_\sigma(P_+)=\cE_\sigma(F)/4$.
  Set $\nu_\beta=\zeta_\beta/2$.
  Equation~\eqref{eq:slow-ckg-bottleneck-inputs} then gives
  \begin{align}\label{eq:warm-refinement-small-crossing}
    p_t\leq e^{-\nu_\beta n}
    \qquad\text{for }0\leq t\leq e^{\nu_\beta n},
  \end{align}
  once $n$ is sufficiently large.

  Let $\mathcal M_t(\rho):=\Tr(\rho E_t)\ket{1}\bra{1}+\Tr(\rho(I-E_t))\ket{0}\bra{0}$ be the channel implementing the binary measurement $\{E_t,I-E_t\}$.
  Relative-entropy data processing yields
  \begin{align}
    D_0\geq D(\tau\|\sigma) \geq D\bigl(\mathcal M_t(\tau)\|\mathcal M_t(\sigma)\bigr) =q_t\log\frac{q_t}{p_t}
       +(1-q_t)\log\frac{1-q_t}{1-p_t}\geq q_t\log\frac{1}{p_t}-\log2.
  \end{align}
  Thus, uniformly over the time interval in~\eqref{eq:warm-refinement-small-crossing},
  \begin{align}\label{eq:warm-refinement-initial-crossing}
    q_t\leq\frac{D_0+\log2}{\nu_\beta n}.
  \end{align}
  If $p_t=0$, faithfulness of $\sigma$ implies $E_t=0$ and hence $q_t=0$, so the same bound holds.
  Combining this with the preceding decomposition of $\Tr\bigl(\tau\mathcal T_t(P_-)\bigr)$ proves~\eqref{eqn:negative-phase-population} and completes the proof.
\end{proof}

\subsection{Technical lemmas for slow mixing}\label{sec:slow-mixing-estimates}

The proof of~\cref{thm:slow-ckg-mixing} uses two technical estimates.
The first bounds the Dirichlet form directly through its sum-of-squares representation, using boundary suppression and local energy-transfer estimates.
The second gives a lower bound on the KMS variance of the centered phase-sign observable.
All matrix norms without subscripts are operator norms, and $\norm{\cdot}_2$ denotes the Hilbert--Schmidt norm.

\begin{lem}[Direct Dirichlet-form bound from the sum-of-squares representation]\label{lem:slow-filtered-phase-flux}
  Consider the Hamiltonian~\eqref{eq:warm-start-hamiltonian} with $0<J_{i,j}\leq\kappa/n$ and the normalized Gaussian CKG sampler, with fixed $\beta\geq3$ and $\beta\kappa\leq1/20$.
  For odd $n\geq17$, let $F=\operatorname{sgn}(M)$ and define
  \begin{align}
    z_n&:=(8\kappa)^{\lfloor n/8\rfloor-1},
      \qquad r_n:=e^{-(\beta/4-\log2)n},\\
    C_n&:=2z_n+4e^{-n/80}+2\sqrt{r_n}.
  \end{align}
  Then
  \begin{align}\label{eq:slow-sos-dirichlet-bound}
    \cE_\sigma(F)\leq e^{-1/4}\left[C_n^2+
      \frac54(\beta\kappa n)^2z_n^2\right].
  \end{align}
  In particular, there are constants $\eta_\beta>0$ and $n_1(\beta)$, independent of the individual allowed couplings, such that
  \begin{align}\label{eq:slow-dirichlet-exponential}
    \cE_\sigma(F)\leq e^{-\eta_\beta n}
  \end{align}
  for every odd $n\geq n_1(\beta)$.
\end{lem}

\begin{proof}
  We first bound real-time commutators with $F$ and then apply the sum-of-squares representation.
  Shift the Hamiltonian by a scalar and define
  \begin{align}
    K&:=H+(n/2+1)I=D_*+V,
      \qquad D_*:=\frac{n^2-M^2}{2n}+1-\frac Mn,\\
    \Pi_u&:=\mathbf{1}_{\{K\leq n/4\}},
      \qquad\Pi_v:=\mathbf{1}_{\{K\leq5n/16\}}.
  \end{align}
  This shift leaves both the Gibbs state and the sampler unchanged.
  Put $\ell_n=\lfloor n/8\rfloor-1$ and
  $P_\ell=\mathbf{1}_{\{|M|\leq n/2-4\ell\}}$ for $0\leq\ell\leq\ell_n$.
  On each of these windows, $D_*\geq3n/8$.
  Since a pair flip changes $M$ by at most four,
  $P_\ell V=P_\ell VP_{\ell-1}$ for $\ell\geq1$.
  Let $D_\ell$ and $K_v$ be the restrictions of $D_*$ and $K$ to the ranges of $P_\ell$ and $\Pi_v$, respectively.
  The map $X_\ell=P_\ell\Pi_v$ satisfies
  \begin{align}
    D_\ell X_\ell-X_\ell K_v=-P_\ell VP_{\ell-1}\Pi_v.
  \end{align}
  The spectral separation $3n/8-5n/16=n/16$ gives the integral representation
  \begin{align}
    X_\ell=-\int_0^\infty e^{-sD_\ell}
      P_\ell VP_{\ell-1}\Pi_v e^{sK_v}\,\d s.
  \end{align}
  Thus $\norm{P_\ell\Pi_v}\leq(16/n)\norm{V}\norm{P_{\ell-1}\Pi_v}
  \leq8\kappa\norm{P_{\ell-1}\Pi_v}$.
  Iteration from $\norm{P_0\Pi_v}\leq1$ and $P_{\rm b}\leq P_{\ell_n}$ yields $\norm{P_{\rm b}\Pi_v}\leq z_n$.
  Moreover, the all-up state has $\bra{0^{\otimes n}}K\ket{0^{\otimes n}}=0$, so Jensen's inequality in its spectral measure gives $\Tr(e^{-\beta K})\geq1$.
  There are at most $2^n$ eigenvalues above $n/4$, hence
  \begin{align}\label{eq:slow-sos-boundary}
    \norm{P_{\rm b}\Pi_v}\leq z_n,
    \qquad\Tr\bigl(\sigma(I-\Pi_u)\bigr)\leq r_n,
    \qquad\Tr(\sigma P_{\rm b})\leq z_n^2+r_n.
  \end{align}
  The last bound follows by splitting the Gibbs state at $\Pi_u\leq\Pi_v$.

  Up to a scalar, $K$ is two-local with total interaction norm incident on any site at most $g_{\rm loc}\leq1+\kappa\leq61/60$.
  For a single-site Pauli $S$, expanding nested commutators gives
  $\norm{\operatorname{ad}_K^k(S)}\leq(2g_{\rm loc})^k k!$:
  after $k-1$ commutators each term has support on at most $k$ sites, and Hamiltonian terms meeting that support have total norm at most $kg_{\rm loc}$.
  At $s=1/5$, the commutator series therefore implies
  $\norm{e^{sK}Se^{-sK}}\leq(1-2g_{\rm loc}s)^{-1}<2$.
  Inserting the two spectral projections, whose energies differ by at least $n/16$, yields
  \begin{align}\label{eq:slow-sos-energy-tail}
    \norm{(I-\Pi_v)S\Pi_u}\leq2e^{-n/80}.
  \end{align}

  Let $F_u=e^{-iuK}Fe^{iuK}$ and $S(u)=e^{iuK}Se^{-iuK}$.
  The commutators $[K,F]$ and $[F,S]$ have both endpoints in $P_{\rm b}$, and $\norm{[K,F]}\leq2\norm{V}\leq\kappa n$.
  Duhamel's formula and~\eqref{eq:slow-sos-boundary} give
  $\norm{(F_u-F)\Pi_v}\leq\kappa n|u|z_n$.
  Expanding $[F_u,S]=[F,S]+(F_u-F)S-S(F_u-F)$, splitting $S\Pi_u$ at $\Pi_v$, and using~\eqref{eq:slow-sos-energy-tail} give
  \begin{align}
    \norm{[F_u,S]\Pi_u}\leq2(1+\kappa n|u|)z_n+4e^{-n/80}.
  \end{align}
  The complementary Gibbs sector contributes at most $2\sqrt{r_n}$ in Hilbert--Schmidt norm.
  Since $K$ commutes with $\sigma$, we obtain
  \begin{align}\label{eq:slow-sos-real-time}
    \norm{[S(u),F]\sqrt\sigma}_2
      \leq C_n+2\kappa n|u|z_n.
  \end{align}

  We now use the sum-of-squares representation, writing $\norm{X}_\sigma^2=\langle X,X\rangle_\sigma$.
  For the CKG sampler, set
  \begin{align}
    g_\beta(t)&:=\frac{1}{\beta\cosh(2\pi t/\beta)},
      \qquad B_{j,\omega}:=\sigma^{-1/4}\widehat A_j(\omega)\sigma^{1/4},\\
    B_{j,\omega}(t)&:=e^{itH}B_{j,\omega}e^{-itH}.
  \end{align}
  The representation specializes to
  \begin{align}\label{eq:slow-ckg-sos}
    \cE_\sigma(F)=\sum_{j=1}^{3n}\int_{\mathbb R}\!\int_{\mathbb R}
      g_\beta(t)\gamma_\beta(\omega)
      \norm{[B_{j,\omega}(t),F]}_\sigma^2\,\d\omega\,\d t.
  \end{align}
  Completing the square in the Gaussian frequency filter gives
  \begin{align}
    B_{j,\omega}
      &=e^{\beta\omega/4+1/16}\widehat A_j\!\left(\omega+\frac{1}{2\beta}\right),\\
    \gamma_\beta(\omega)e^{\beta\omega/2+1/8}
      &=e^{-1/4}e^{-\beta^2\xi^2/2},
      \qquad\xi:=\omega+\frac{1}{2\beta}.
  \end{align}
  For any operator $X$, the arithmetic--geometric mean inequality in a $\sigma$-eigenbasis gives
  \begin{align}
    \norm{X}_\sigma^2\leq\frac12\left(
      \norm{X\sqrt\sigma}_2^2+\norm{X^\dagger\sqrt\sigma}_2^2\right).
  \end{align}
  Write $\widehat A_{j,t}(\xi)=e^{itH}\widehat A_j(\xi)e^{-itH}$.
  Since the couplings are Hermitian and $f_\beta$ is real,
  $\widehat A_{j,t}(\xi)^\dagger=\widehat A_{j,t}(-\xi)$.
  The two Hilbert--Schmidt terms thus have equal integrals against the even weight $e^{-\beta^2\xi^2/2}$.
  Dropping this weight and applying Parseval's identity yields
  \begin{align}
    \cE_\sigma(F)
      &\leq e^{-1/4}\sum_j\int_{\mathbb R}\!\int_{\mathbb R}
        g_\beta(t)|f_\beta(s)|^2
        \norm{[A_j(t+s),F]\sqrt\sigma}_2^2\,\d s\,\d t\nonumber\\
      &\leq e^{-1/4}\int_{\mathbb R}\!\int_{\mathbb R}
        g_\beta(t)|f_\beta(s)|^2
        \bigl(C_n+2\kappa n|t+s|z_n\bigr)^2\,\d s\,\d t,
  \end{align}
  where $A_j(u)=e^{iuH}A_je^{-iuH}$ and the last step uses~\eqref{eq:slow-sos-real-time} and $A_j=S_j/\sqrt{3n}$.
  Both time weights are even, with moments
  \begin{align}
    \int_{\mathbb R}g_\beta(t)\,\d t=\frac12,
    \quad\int_{\mathbb R}t^2g_\beta(t)\,\d t=\frac{\beta^2}{32},
    \quad\int_{\mathbb R}|f_\beta(s)|^2\,\d s=1,
    \quad\int_{\mathbb R}s^2|f_\beta(s)|^2\,\d s=\frac{\beta^2}{4}.
  \end{align}
  Using $(x+y)^2\leq2x^2+2y^2$ therefore proves~\eqref{eq:slow-sos-dirichlet-bound}.
  Finally, $z_n\leq(2/15)^{\lfloor n/8\rfloor-1}$ and $r_n\leq e^{-(3/4-\log2)n}$, so all terms in that bound decay exponentially after absorbing polynomial prefactors.
  This gives~\eqref{eq:slow-dirichlet-exponential} with a size threshold uniform over the allowed couplings.
\end{proof}

\begin{lem}[Nonvanishing KMS variance of the phase sign]\label{lem:slow-kms-variance}
  For the Hamiltonian~\eqref{eq:warm-start-hamiltonian}, let $n$ be odd, $\beta>0$, and $J_{i,j}>0$ for $i\ne j$.
  Define
  \begin{align}
    F:=\operatorname{sgn}(M),\qquad F_0:=F-\Tr(\sigma F)I,
    \qquad P_\pm:=\frac{I\pm F}{2},\qquad b_\beta:=\frac12e^{-2\beta}.
  \end{align}
  Both phases satisfy $\Tr(\sigma P_\pm)\geq b_\beta$, and
  \begin{align}\label{eq:slow-kms-variance}
    \langle F_0,F_0\rangle_\sigma
    &=1-\bigl(\Tr(\sigma F)\bigr)^2-\frac12\norm{[\sqrt\sigma,F]}_2^2\\
    &\geq4b_\beta(1-b_\beta)-\frac12\norm{[\sqrt\sigma,F]}_2^2.\nonumber
  \end{align}
  Consequently, under the coupling and temperature assumptions of~\cref{lem:slow-filtered-phase-flux}, for all sufficiently large odd $n$,
  \begin{align}\label{eq:slow-kms-variance-lower}
    \langle F_0,F_0\rangle_\sigma\geq2b_\beta(1-b_\beta)>0,
  \end{align}
  with the size threshold uniform over the allowed couplings.
\end{lem}

\begin{proof}
  Remove the longitudinal tilt and write
  $H_0=-M^2/(2n)+V$, $Z_0=\Tr(e^{-\beta H_0})$, and $\sigma_0=e^{-\beta H_0}/Z_0$.
  The global spin flip $U=\prod_{i=1}^n X_i$ commutes with $H_0$ and satisfies $UP_\pm U^\dagger=P_\mp$.
  Since $n$ is odd, $P_++P_-=I$, so $\Tr(\sigma_0P_\pm)=1/2$.

  To compare these weights with those of $\sigma$, set $D_0=-M^2/(2n)$ and $D=D_0-M/n$ and use the Trotter products
  \begin{align}
    e^{-\beta H}=\lim_{N\to\infty}
      \bigl(e^{-\beta D/N}e^{-\beta V/N}\bigr)^N,
    \qquad
    e^{-\beta H_0}=\lim_{N\to\infty}
      \bigl(e^{-\beta D_0/N}e^{-\beta V/N}\bigr)^N.
  \end{align}
  All factors have nonnegative entries in the computational basis: the diagonal factors are positive, and $-V=\sum_{i<j}J_{i,j}X_iX_j$ is entrywise nonnegative.
  Replacing $D_0$ by $D$ multiplies each diagonal factor along a Trotter path by $e^{\beta M/(nN)}\in[e^{-\beta/N},e^{\beta/N}]$.
  Hence each complete path weight changes by a factor in $[e^{-\beta},e^\beta]$.
  Summing these nonnegative weights and taking the Trotter limit gives
  \begin{align}
    \Tr(P_\pm e^{-\beta H})
      &\geq e^{-\beta}\Tr(P_\pm e^{-\beta H_0})
      =\frac12e^{-\beta}Z_0,\\
    Z_H&\leq e^\beta Z_0.
  \end{align}
  Thus $p_\pm:=\Tr(\sigma P_\pm)\geq\frac12e^{-2\beta}=b_\beta$.

  Put $\mu=\Tr(\sigma F)=p_+-p_-$ and $w_n=\norm{[\sqrt\sigma,F]}_2$.
  Using $F^2=I$, direct expansion gives
  \begin{align}
    w_n^2=2-2\Tr(\sqrt\sigma F\sqrt\sigma F)
      =2-2\langle F,F\rangle_\sigma.
  \end{align}
  Since $\langle I,F\rangle_\sigma=\mu$ and $\langle I,I\rangle_\sigma=1$, it follows that
  \begin{align}
    \langle F_0,F_0\rangle_\sigma =1-\mu^2-\frac12w_n^2
      =4p_+p_--\frac12w_n^2 \geq4b_\beta(1-b_\beta)-\frac12w_n^2,
  \end{align}
  where we used $p_++p_-=1$ and $p_\pm\geq b_\beta$.
  This proves~\eqref{eq:slow-kms-variance} without any smallness assumption on the couplings.

  For the final assertion, impose the assumptions of~\cref{lem:slow-filtered-phase-flux}.
  In an energy eigenbasis, the Gibbs eigenvalues $\lambda_r=e^{-\beta E_r}/Z_H$ obey
  \begin{align}
    (\sqrt{\lambda_r}-\sqrt{\lambda_s})^2
      \leq\frac{\beta^2}{8}(E_r-E_s)^2(\lambda_r+\lambda_s).
  \end{align}
  Summing against $\abs{F_{rs}}^2$ and using $F=F^\dagger$ yields
  $w_n\leq\frac{\beta}{2}\norm{[H,F]\sqrt\sigma}_2$.
  Phase-changing matrix elements of $[H,F]$ have both endpoints in $P_{\rm b}$, and $\norm{[H,F]}\leq2\norm{V}\leq\kappa n$.
  The boundary estimate~\eqref{eq:slow-sos-boundary} therefore gives
  \begin{align}
    w_n\leq\frac{\beta\kappa n}{2}\sqrt{\Tr(\sigma P_{\rm b})}
      \leq\frac{\beta\kappa n}{2}\sqrt{z_n^2+r_n}.
  \end{align}
  The right-hand side tends to zero exponentially, uniformly over the allowed couplings.
  For each fixed $\beta$, it follows that $w_n^2/2\leq2b_\beta(1-b_\beta)$ for all sufficiently large odd $n$.
  Substituting into~\eqref{eq:slow-kms-variance} proves~\eqref{eq:slow-kms-variance-lower} with a uniform size threshold.
\end{proof}

\section{Computing observables using purified Gibbs states}
\label{append:computing-observables}

Once the purified Gibbs state is available, estimating thermal observables becomes a standard expectation-value problem on the doubled Hilbert space. For any Hermitian observable $O$,
\begin{equation}
\Tr(\sigma O)=\langle\!\langle \sigma^{1/2}|I\otimes O|\sigma^{1/2}\rangle\!\rangle.
\label{eq:observable-estimation}
\end{equation}
Thus the output of the QSVT routine is already the coherent object needed for physics applications: thermal expectation values are obtained by probing one register of the purification rather than by introducing a separate sampling primitive.

If a block-encoding of $O$ is available, standard amplitude-estimation methods~\cite{brassard2000quantum,rall2021faster,rall2023amplitude} convert \cref{eq:observable-estimation} into an additive-$\epsilon$ estimation procedure. 
The cost is the state-preparation cost $\Or\!\left(\frac{1}{\sqrt{\Delta}}\log\frac{1}{\epsilon}\right)$, with a $1/\epsilon$ multiplicative overhead for expectation estimation.

\end{document}